\documentclass[reqno]{article}
\pdfoutput=1

\usepackage[utf8]{inputenc}
\usepackage[T1]{fontenc}
\usepackage{microtype}
\usepackage[noend]{algorithmic}
\usepackage{amsmath, amssymb, amsthm}
\usepackage{thmtools}
\usepackage{thm-restate}
\usepackage{blindtext}
\usepackage{hyperref} % Must load hyperref before cleverref
\usepackage{cleveref}
\usepackage{comment}
\usepackage{enumerate}
\usepackage{mathtools}
\usepackage{subcaption}
\usepackage[dvipsnames]{xcolor}
\usepackage{dsfont}
\usepackage{xspace}

\usepackage{microtype}
\usepackage{graphicx}
\usepackage{booktabs} % for professional tables

\usepackage{afterpage}

\crefformat{equation}{(#2#1#3)}

\DeclareMathOperator*{\argmax}{arg\,max}

\newcommand{\Prob}[2]{\Pr_{#1}\parens*{#2}}

\newcommand{\E}{\mathbb{E}}

\newcommand{\Exp}[2]{\E_{#1}\bracks*{#2}}
\newcommand{\ExpCond}[3]{\E_{#1}\bracks*{#2 \;\middle|\; #3}}
\newcommand{\AdaptiveSampling}{\textsc{Threshold-Sampling}\xspace}
\newcommand{\NMTS}{\textsc{Non-Monotone-Threshold-Sampling}\xspace}
\newcommand{\NonmonotoneMaximization}{\textsc{Adaptive-Nonmonotone-Max}\xspace}
\newcommand{\UnconstrainedMaximization}{\textsc{Unconstrained-Max}\xspace}

\newcommand{\EstimateMean}{\textsc{Reduced-Mean}\xspace}

\newcommand{\RandomLazyGreedy}{\textsc{Random-Lazy-Greedy-Improved}\xspace}
\newcommand{\Blits}{\textsc{Blits}\xspace}
\newcommand{\Fantom}{\textsc{Fantom}\xspace}
\newcommand{\Greedy}{\textsc{Greedy}\xspace}
\newcommand{\Random}{\textsc{Random}\xspace}
\newcommand{\DEF}{\stackrel{\textnormal{\tiny\sffamily def}}{=}}

\newcommand{\R}{\mathbb{R}}
\newcommand{\OPT}{\textnormal{OPT}}
\newcommand{\ALG}{\textnormal{ALG}}
\newcommand{\poly}{\textnormal{poly}}
\newcommand{\Gap}{\textnormal{Gap}}
\newcommand{\todo}[1]{\noindent\textbf{\color{red} TODO: #1}}
\newcommand{\ind}{\mathds{1}}
\newcommand{\UnconstrainedApprox}{\alpha}
\newcommand{\change}[1]{{\color{black}#1}}
\newcommand{\morteza}[1]{{\color{black}#1}}

 \newcommand{\cD}{\mathcal{D}}

\DeclarePairedDelimiter{\abs}{\lvert}{\rvert}
\DeclarePairedDelimiter{\set}{\{}{\}}
\DeclarePairedDelimiter{\parens}{(}{)}
\DeclarePairedDelimiter{\bracks}{[}{]}
\DeclarePairedDelimiter{\floor}{\lfloor}{\rfloor}
\DeclarePairedDelimiter{\ceil}{\lceil}{\rceil}

\theoremstyle{plain}
\newtheorem{theorem}{Theorem}[section]
\newtheorem{lemma}[theorem]{Lemma}

\theoremstyle{definition}
\newtheorem{definition}[theorem]{Definition}

\usepackage[accepted]{icml2019}

\icmltitlerunning{
  Non-monotone Submodular Maximization with 
  Nearly Optimal Adaptivity and Query Complexity
}

\begin{document}

\onecolumn{
\icmltitle{
  Non-monotone Submodular Maximization with
  Nearly Optimal Adaptivity and Query Complexity
}

\begin{icmlauthorlist}
\icmlauthor{Matthew Fahrbach}{gt}
\icmlauthor{Vahab Mirrokni}{google}
\icmlauthor{Morteza Zadimoghaddam}{google}
\end{icmlauthorlist}

\icmlaffiliation{gt}{Georgia Institute of Technology}
\icmlaffiliation{google}{Google Research}

\icmlcorrespondingauthor{Matthew Fahrbach}{matthew.fahrbach@gatech.edu}

\vskip 0.3in
}

\printAffiliationsAndNotice{} % otherwise use the standard text.

\begin{abstract}
Submodular maximization is a general optimization problem with a wide range of
applications in machine learning (e.g., active learning, clustering, and
feature selection).
In large-scale optimization, the parallel running time of an algorithm
is governed by its adaptivity, which measures the number of
sequential rounds needed if the algorithm can execute polynomially-many
independent oracle queries in parallel.
While low adaptivity is ideal, it is not sufficient for an
algorithm to be efficient in practice---there are many applications of
distributed submodular optimization where the
number of function evaluations becomes prohibitively expensive.
Motivated by these applications, we study the
adaptivity and query complexity of submodular maximization.
In this paper,
we give the first constant-factor approximation algorithm for
maximizing a non-monotone submodular function subject to a cardinality
constraint~$k$
that runs in $O(\log(n))$ adaptive rounds
and makes $O(n \log(k))$ oracle queries in expectation.
In our empirical study,
we use three real-world applications to compare our algorithm with
several benchmarks for non-monotone submodular maximization.
The results demonstrate that our algorithm finds competitive solutions using
significantly fewer rounds and queries.
\end{abstract}

\section{Introduction}
\label{sec:introduction}

Submodular set functions are a powerful tool for modeling real-world
problems because they naturally exhibit the property of diminishing returns.
Several well-known examples of submodular functions include graph cuts,
entropy-based clustering, coverage functions,
and mutual information.
As a result, submodular functions have been increasingly used in
applications of machine learning such as
data summarization~\cite{simon2007scene,sipos2012temporal,tschiatschek2014learning},
feature selection~\cite{DK08,KEDNG17},
and recommendation systems~\cite{el2011beyond}.
While some of these applications involve maximizing monotone submodular functions,
the more general problem of non-monotone
submodular maximization has also been used
extensively~\cite{feige2011maximizing,buchbinder2014submodular,mirzasoleiman2016fast,balkanski2018non,norouzi2018beyond}.
Some specific applications of non-monotone submodular maximization include
image summarization and movie recommendation~\cite{mirzasoleiman2016fast}, and
revenue maximization in viral marketing~\cite{HMS08}.
Two compelling uses of non-monotone submodular maximization algorithms are:
\begin{itemize}
\item
Optimizing objectives that are a monotone submodular function minus a linear cost
function that penalizes the addition of more elements to the set (e.g., the coverage and
diversity trade-off). This appears in facility location problems where opening
centers is expensive and in exemplar-based clustering~\cite{dueck2007non}.
\item
Expressing learning problems such as feature selection
using weakly submodular functions~\cite{DK08,KEDNG17,elenberg2018restricted,qian2019fast}.
One possible source of non-monotonicity in this context is overfitting to
training data by selecting too many representative
features~(e.g., Section 1.6 and Corollary 3.19 in \citet{mohri2018foundations}).
Although most of these learning problems have not yet been rigorously modeled as
non-monotone submodular functions, there has been a recent surge of interest and
a substantial amount of momentum in this direction.
\end{itemize}

%Submodular functions have the natural property of diminishing returns, making
%them prominent in applied fields such as machine learning and data mining.
%There has been a surge in applying submodular optimization for data
%summarization \cite{tschiatschek2014learning,simon2007scene,sipos2012temporal},
%recommendation systems \cite{el2011beyond}, and 
%feature selection for learning models~\cite{DK08,KEDNG17},
%to name a few applications.  There are also numerous recent works that focus on
%maximizing submodular functions from a theoretical perspective. Depending on
%the setting where the submodular maximization algorithms are applied, new
%challenges emerge and hence more practical algorithms have been designed to
%solve the problem in distributed \cite{nips13,MirrokniZadim2015,barbosa2015power},
%streaming \cite{badanidiyuru2014streaming}, and robust
%\cite{mirzasoleiman2017deletion,mitrovic2017streaming,kazemi2018scalable}
%optimization frameworks.

The literature on submodular optimization typically assumes access to an oracle
that evaluates the submodular function.
In practice, however, oracle queries
may take a long time to process.
For example, the log-determinant of submatrices of a positive semi-definite
matrix is a submodular function that is notoriously expensive to
compute~\cite{kazemi2018scalable}.
Therefore, our goal when designing distributed algorithms is to minimize
the number of rounds where the algorithm communicates with the oracle.
This motivates the notion of the \emph{adaptivity complexity} of submodular optimization,
first investigated in~\citet{BS18}.
In this model of computation, the algorithm can ask polynomially-many
independent oracle queries all together in each round.

In a wide range of machine learning optimization problems, the objective
functions can only be estimated through oracle access to the function.
In many instances, these oracle evaluations are a new time-consuming
optimization problem that we treat as a black box (e.g., hyperparameter optimization).
Since our goal is to optimize the objective function using as
few rounds of interaction with the oracle as possible,
insights and algorithms developed in this adaptivity complexity framework
can have a deep impact on distributed computing for machine learning
applications in practice.
Further motivation for the importance of this computational model is given in~\citet{BS18}.

While the number of adaptive rounds is an important quantity to minimize,
the computational complexity
of evaluating oracle queries also motivates the design of algorithms that are
efficient in terms of the total number of oracle queries.
An algorithm typically needs to make at least a constant number of queries per
element in the ground set to achieve a constant-factor approximation. 
In this paper, we study the adaptivity complexity and the total number of
oracle queries that are needed to guarantee a constant-factor
approximation when maximizing a non-monotone submodular function.

\textbf{Results and Techniques.}
Our main result is a distributed algorithm for maximizing a non-monotone
submodular function subject to a cardinality constraint $k$ that achieves an
expected $(0.039-\varepsilon)$-approximation
in $O(\log(n)/\varepsilon)$ adaptive rounds
using $O(n\log(k)/\varepsilon^2)$ expected function evaluation queries.
To the best of our knowledge, this is the first constant-factor approximation
algorithm with nearly-optimal adaptivity for the general problem of maximizing
non-monotone submodular functions. 
The adaptivity complexity of our algorithm is optimal up to a
$\Theta(\log\log(n))$ factor by the lower bound in~\citet{BS18}.

The building blocks of our algorithm are
\change{(1)} the \AdaptiveSampling subroutine in~\citet{fahrbach2019submodular},
which returns a subset of high-valued elements
in $O(\log(n)/\varepsilon)$ adaptive rounds,
and \change{(2)} the unconstrained submodular maximization algorithm
in~\citet{chen2018unconstrained}, \change{which} gives a $(1/2-\varepsilon)$-approximation
in $O(\log(1/\varepsilon)/\varepsilon)$ adaptive rounds.
We modify \AdaptiveSampling to terminate early if its pool of candidate elements
becomes too small \change{to guarantee} that each element is not chosen with at least constant probability.
This property has been shown to be useful for
obtaining constant-factor approximations for non-monotone
submodular function maximization~\cite{buchbinder2014submodular}.
Next, we run unconstrained maximization on the remaining set of
high-valued candidates if its size is at most $3k$, downsample accordingly,
and output the better of the two solutions.
Our analysis shows how to optimize the constant parameters to balance between
the two behaviors.
Last, since \AdaptiveSampling requires an input close to $\OPT/k$,
we find an interval containing $\OPT$, try logarithmically-many input
thresholds in parallel, and return the solution with maximum value.
We note that improving the bounds for $\OPT$ via low-adaptivity preprocessing
can reduce the total query complexity
as shown in~\citet{fahrbach2019submodular}.

\begin{table*}
\caption{
  Independent and concurrent works 
  for low-adaptivity non-monotone submodular maximization
  subject to a cardinality constraint.}
\label{comparison-table}
%\vskip 0.15in
\begin{center}
\begin{small}
\begin{sc}
  \begin{tabular}{lccc}
\toprule
Algorithm & Approximiation & Adaptivity & Queries \\
\midrule
    \citet{buchbinder2016comparing} & $1/e - \varepsilon$ & $O(k)$ & $O(n)$ \\
  \citet{balkanski2018non} & $1/(2e) - \varepsilon $ & $O(\log^2(n))$ & $O(\OPT^2 n \log^2(n) \log(k) )$ \\
    \citet{chekuri2018parallelizing} & $3-2\sqrt{2}-\varepsilon$ & $O(\log^2(n))$ & $O(nk^4 \log^2(n))$\\
    \citet{ene2018submodular} & $1/e-\varepsilon$ & $O(\log^2(n))$ & $O(n k^2 \log^2(n))$\\
  This paper & $0.039 - \varepsilon$ & $O(\log(n))$ & $O(n\log(k))$ \\
    \change{\citet{amanatidis2021submodular}} & \change{$0.171 - \varepsilon$} & \change{$O(\log(n))$} & \change{$O(nk \log(n)\log(k))$} \\
   \change{\citet{chen2022practical}} & \change{$1/6 - \varepsilon$} & \change{$O(\log(n))$} & \change{$O(n\log(k))$} \\
   \change{\citet{chen2022practical}} & \change{$0.193 - \varepsilon$} & \change{$O(\log^2(n))$} & \change{$O(n\log(k))$} \\
\bottomrule
\end{tabular}
\end{sc}
\end{small}
\end{center}
%\vskip -0.1in
\end{table*}

\textbf{Related Works.}
Submodular maximization has garnered a significant amount of attention in the
distributed and streaming literature because of its role in large-scale data
mining
\cite{spaa-LMSV11,nips13,badanidiyuru2014streaming,KMVV13, MirrokniZadim2015,
barbosa2015power,alina2,liu2018submodular}.
However, in many distributed models (e.g., the Massively Parallel Computation
model), round complexity often captures a different notion than adaptivity
complexity.
For example, a constant-factor approximation is achievable in two rounds of
computation~\cite{MirrokniZadim2015},
but it is impossible to compute a constant-factor approximation in
$O(\log(n)/\log\log(n))$ adaptive rounds~\cite{BS18}.
Since adaptivity measures the communication complexity with a function
evaluation oracle, a round in most distributed models can have arbitrarily high
adaptivity.

The first set of related works with low adaptivity focus on maximizing
\emph{monotone} submodular functions subject to a cardinality constraint~$k$.
In \citet{BS18}, the authors show that a $(1/3-\varepsilon)$-approximation is
achievable in $O(\log(n))$ rounds.
In terms of parallel running time, this is exponentially faster than
the celebrated greedy algorithm which gives a $(1-1/e)$-approximation
in $O(k)$ rounds~\cite{Nemhauser_Wolsey_Fisher78}.
Subsequently, \citet{BRS18,EN18,fahrbach2019submodular} independently
designed $(1-1/e-\varepsilon)$-approximation algorithms with $O(\log(n))$
adaptivity. These works also show that only $O(n)$ oracle queries are
needed in expectation.
Recent works have also investigated the adaptivity of the multilinear
relaxation of monotone submodular functions subject to packing
constraints~\cite{chekuri2019submodular} and
the submodular cover problem~\cite{agarwal2019stochastic}.

%Substantially less progress has been made on the more general problem of
%maximizing a (not necessarily monotone) submodular function.
%The best achievable approximation is unknown but in the range
%$[0.385,0.491]$~\cite{buchbinder2016constrained,gharan2011submodular}.
%Recently, \cite{balkanski2018non} designed a
%parallel algorithm for non-monotone submodular maximization subject to
%a cardinality constraint that gives a $(1/(2e)-\varepsilon)$-approximation
%in $O(\log^2(n))$ adaptive rounds.
%Since this algorithm estimates the expected marginal gains of random
%subsets, the number of function evaluations it needs to achieve
%provable guarantees is $O(\OPT^2 n \log^2(n)\log(k))$.
%We acknowledge that this can likely be improved by normalization or
%estimating an indicator random variable instead.
%The more recent works of~\cite{chekuri2018parallelizing,ene2018submodular}
%give constant-approximation algorithms with $O(\log^2(n))$ adaptivity
%for maximizing non-monotone submodular functions subject to matroid constraints.
%These approaches are based on multilinear extensions
%and require $O(n \cdot \poly(k, \log(n)))$ function evaluations 
%to simulate an oracle for $\nabla f$ with high enough accuracy.
%Lastly, there have also been significant advancements in unconstrained
%submodular maximization with low adaptivity~\cite{chen2018unconstrained,ene2018parallel}.

While the general problem of maximizing a (not necessarily monotone)
submodular function has been studied extensively
\cite{LMNS10,feige2011maximizing,gharan2011submodular,buchbinder2014submodular},
noticeably less progress has been made.
For example, the best achievable approximation for the centralized maximization
problem is unknown but in the range
$[0.385,0.491]$~\cite{buchbinder2016constrained,gharan2011submodular}.
However, some progress has been made for the adaptive complexity of this
problem, all which has been done independently and concurrently with an
earlier version of this paper.
Recently, \citet{balkanski2018non} designed a
parallel algorithm for non-monotone submodular maximization subject to
a cardinality constraint that gives a $(1/(2e)-\varepsilon)$-approximation
in~$O(\log^2(n))$ adaptive rounds.
Their algorithm estimates the expected marginal gain of random
subsets, and therefore the number of function evaluations it needs to achieve
provable guarantees is $O(\OPT^2 n \log^2(n)\log(k))$.
We acknowledge that the query complexity can likely be improved via
normalization or estimating an indicator random variable instead. 
The works of~\citet{chekuri2018parallelizing} and \citet{ene2018submodular}
give constant-factor approximation algorithms with $O(\log^2(n))$ adaptivity
for maximizing non-monotone submodular functions subject to matroid constraints.
Their approaches use multilinear extensions
and require $\Omega(n k^2 \log^2(n))$ function evaluations 
to simulate an oracle for $\nabla f$ with high enough accuracy.
There have also been significant advancements in low-adaptivity algorithms
for the problem of unconstrained submodular
maximization~\cite{chen2018unconstrained,ene2018parallel}.

\change{
In a subsequent work,
\citet{kuhnle2021nearly} improved on our $(0.039-\varepsilon)$-approximation ratio
and gave a $(1/6 - \varepsilon)$-approximation algorithm
that uses $O(\log(n))$ adaptive rounds and $O(n \log (k))$ queries.
Furthermore, he also designed an algorithm that achieves a
$(0.193-\varepsilon)$-approximation
in $O(\log^2(n))$ adaptive rounds
using $O(n \log (k))$ queries.
Both our original algorithm in the ICML 2019 manuscript
and the two algorithms in \citet{kuhnle2021nearly}, however,
relied on the \textsc{Threshold-Sampling} subroutine in~\citet{fahrbach2019submodular},
whose guarantees only hold for \emph{monotone} submodular functions.
After this was brought to the authors' attention,
our updated manuscript and the new work of~\citet{chen2022practical}
independently fix the threshold sampling bug with a slight change,
so all of these low-adaptivity non-monotone submodular results continue to hold.
}
\section{Preliminaries}
\label{sec:preliminaries}

For any set function $f : 2^N \rightarrow \R_{\ge 0}$
and subsets ${S,T \subseteq N}$,
let $\Delta(T, S) \DEF f(S \cup T) - f(S)$ denote the \emph{marginal gain} of~$f$
at $T$ with respect to $S$.
We refer to $N$ as the ground set and let $|N| = n$.
A set function ${f : 2^N \rightarrow \R_{\ge 0}}$ is \emph{submodular} if for all
$S \subseteq T \subseteq N$ and any $x \in N \setminus T$ we have
$\Delta(x, S) \ge \Delta(x, T)$, where the marginal gain
notation is overloaded for singletons.
A set function is \emph{monotone} if for all subsets $S \subseteq T \subseteq N$
we have $f(S) \le f(T)$.
In this paper, we investigate distributed algorithms for maximizing submodular
functions subject to a cardinality constraint, including
those that are \emph{non-monotone}.
Let $S^*$ be a solution set to the 
maximization problem $\max_{S \subseteq N} f(S)$ subject to
the cardinality constraint $|S| \le k$, and
let $\mathcal{U}(A, t)$ denote the uniform distribution over all subsets of $A$
of size $t$. 

Our algorithms take as input an \emph{evaluation oracle} for $f$,
which for every query $S \subseteq N$ returns $f(S)$ in $O(1)$ time.
Given an evaluation oracle, we define the \emph{adaptivity} of an algorithm
to be the minimum number of rounds needed such that in each round the algorithm
makes $O(\poly(n))$ independent queries to the evaluation oracle.
Queries in a given round may depend on the answers of queries from previous
rounds but not the current round.
We measure the parallel running time of an algorithm by its
adaptivity.

One of the inspirations for our algorithm is the following lemma, which is
remarkably useful for achieving a constant-factor approximation for
general submodular functions.

\begin{lemma}{\cite{buchbinder2014submodular}}
\label{lem:empty-lower-bound}
Let ${f : 2^N \rightarrow \R_{\ge 0}}$ be submodular. Denote by $A(p)$ a
random subset of $A$ where each element appears with probability at most $p$
(not necessarily independently). Then,
$\E[f(A(p))] \ge (1-p)f(\emptyset)$.
\end{lemma}

In our case, if $S$ is the output of the algorithm and the probability of any element
appearing in~$S$ is bounded away from~$1$, we can
analyze the submodular function
$g : 2^N \rightarrow \R_{\ge 0}$ defined by
$g(S) = f(S \cup S^*)$
to lower bound $\Exp{}{f(S \cup S^*)}$ in terms of $\OPT = f(S^*)$
since $g(\emptyset) = f(S^*)$.

\subsection{\change{Threshold-Sampling Algorithms}}
\label{sec:threshold}

%\todo{Make clear that we introduce a new algorithm \NMTS}
We start with a high-level description of the \AdaptiveSampling algorithm
in~\citet{fahrbach2019submodular}. \change{This algorithm was originally designed for \emph{monotone} submodular functions,
but after a small change can become
the main subroutine of our \emph{non-monotone} maximization algorithm.}
For a threshold~$\tau$, \AdaptiveSampling iteratively builds a solution
set~$S$ over $O(\log(n)/\varepsilon)$ adaptive rounds
and maintains a set of unchosen candidate elements~$A$.
Initially, the solution set is empty and all elements
are candidates (i.e., $S = \emptyset$ and $A = N$).
In each round, the algorithm starts by discarding elements in $A$ whose
marginal gain to the current solution~$S$ is less than the threshold $\tau$.
\change{Then the algorithm efficiently finds the largest cardinality $t^*$
such that if we sample $t^*$ elements from $A$ uniformly at random and add them to $S$ 
in a random order, each addition yields at least a marginal value of $\tau$ with probability at least $1 - \varepsilon$.

For any two sets $A$ and $S$, if the elements in $A$ are added to $S$ in random order, the probability that the $t$-th addition gives a marginal value of at least $\tau$ is a non-increasing function in $t$ by submodularity \citep[Lemma 3.4]{fahrbach2019submodular}. Thus, the notion of $t^*$ is well-defined.
Estimating this probability for any value of $t$ can be done with a few samples and hence an efficient number of oracle queries. 
Trying all values of $t$, however, increases the query complexity drastically, so we compute an estimate $\hat{t} = \floor{(1+\varepsilon)^i} \approx t^*$ by trying a geometrically increasing series of values for $t$. 
At the end of each round, the algorithm samples a set $T \sim U(A, \hat{t})$
and updates the current solution to be $S \cup T$.}

%Then the algorithm efficiently finds the largest cardinality $t^*$
%such that for $T \sim U(A, t^*)$ uniformly at random we have
%$\Exp{}{\Delta(T,S) / |T|} \ge (1-\varepsilon)\tau$.
%At the end of a round, the algorithm samples $T \sim U(A, t^*)$ and
%updates the current solution to be $S \cup T$.

The random choice of $T$ in \AdaptiveSampling has two beneficial effects.
\change{First, it ensures that each added element has marginal value at least 
$\tau$ with probability at least $1-\varepsilon$.
Second, the maximality of $\hat{t}$ implies that an expected $\varepsilon$-fraction 
of candidates are filtered out of $A$ in the next round.
}
%First, it ensures that in expectation
%the average contribution of each element in the 
%returned set is at least $(1-\varepsilon)\tau$.
%Second, it implies that an expected $\varepsilon$-fraction 
%of candidates are filtered out of $A$ in each round.
Therefore, the number of elements that the algorithm considers in
each round decreases geometrically. It follows that
$O(\log(n)/\varepsilon)$ rounds suffice to guarantee 
that when the algorithm terminates, we either have $|S| = k$ or
the marginal gains of all the elements are below the threshold
\change{with high probability.}

\change{
For a monotone function, this approach allows us to achieve the desired approximation guarantee because most,
i.e., at least a $(1-\varepsilon)$-fraction,
of the added elements provide a marginal value of at least $\tau$.
Further, the remaining $\varepsilon$-fraction of elements in $S$
cannot decrease the value of the set because of monotonocity.

\paragraph{Changes for non-monotone functions.}
For non-monotone functions, however,
adding even a single element can degrade the set value substantially.
To guard against this, we introduce an extra filtering step
at the end of each threshold sampling round
to return a subset $S' \subseteq S$ of elements such that
each element provides the required marginal value of $\tau$.
We call this modified version with the \emph{post-filtering steps} the
\NMTS algorithm.
Concretely, since we use an approximate maximum threshold
$\hat{t} \approx t^*$, the $\varepsilon$-fraction of elements added to $S$
after the true threshold $t^*$ can have negative marginal values,
which lowers the solution quality and prevents the use of Markov's inequality
in the analysis.\footnote{\change{The post-filtering steps in \NMTS fixes an error in the original ICML 2019 manuscript.}}
}

% \todo{
% \begin{itemize}
%     \item Discuss $S'$ and more importantly where we depart from monotonicity: we get $(1-\varepsilon)\tau$ gain w.p. at least $1-\varepsilon$, but with $\epsilon$ we can go very negative.
% \end{itemize}
% }

Before presenting \change{\NMTS in~\Cref{alg:sampling}}, we define the distribution
$\cD_t$ from which \change{this algorithm} samples when estimating the maximum
cardinality~$t^*$. Sampling from this Bernoulli distribution can be simulated
with two calls to the evaluation oracle.
\change{This step corresponds to calling the \textsc{Reduced-Mean}
estimator subroutine in \citet{fahrbach2019submodular}.}

\begin{definition}
\label{def:indicator_distribution}
Conditioned on the current state of the algorithm,
consider the process where \change{we sample $T \sim \mathcal{U}(A, t)$}
and then $x \sim A\setminus T$ \change{uniformly at random}.
\change{For all $t \in \{0,1,\dots,|A|-1\}$,}
let $\mathcal{D}_t$ denote the \change{Bernoulli} distribution over the indicator random variable
\begin{align*}
I_t \DEF \ind\bracks*{\Delta(x,S\cup T) \ge \tau}.
\end{align*}
\change{For completeness, we define $I_{|A|} = 0$.}
\end{definition}

\noindent
\change{It is useful to think of}
$\Exp{}{I_t}$ as the probability that the \change{$(t+1)$-st
marginal gain} is at least threshold~$\tau$ if the candidates in $A$ are inserted
into $S$ according to a \change{uniformly} random permutation.

\begin{algorithm}[H]
  \caption{\change{\NMTS}}
  \label{alg:sampling}
  \textbf{Input:} oracle for $f : 2^N \rightarrow \R_{\ge 0}$, constraint $k$,
    threshold $\tau$,
    \change{candidates scale factor $c \ge 1$,} error~$\varepsilon$, failure probability~$\delta$

  \begin{algorithmic}[1]
    \STATE Set smaller error $\hat{\varepsilon} \leftarrow \varepsilon/3$
    \STATE Set 
        $r \leftarrow \ceil{\log_{(1-\hat\varepsilon)^{-1}}(2n/\delta)}$,
        $m \leftarrow \ceil{\log_{\change{(1+\hat\varepsilon)}}(k)}$
    \STATE Set smaller failure probability
        $\hat{\delta} \leftarrow \delta/(2r(m+1))$
    \STATE Initialize \change{$S' \gets \emptyset$}, $S \leftarrow \emptyset$, $A \leftarrow N$
    \FOR{$r$ rounds}
      \STATE Filter $A \leftarrow \{x \in A : \Delta(x,S) \ge \tau\}$ \label{step:inner_filter}
      \IF{$\abs{A}\ {\change{< ck}}$}
        \STATE \textbf{break}
      \ENDIF
      \FOR{$i=0$ to $m$} 
        \STATE Set $t \leftarrow \min\{\floor{(1 + \hat{\varepsilon})^i}, |A|\}$
        \IF{$\EstimateMean(\mathcal{D}_t, \hat{\varepsilon}, \hat{\delta})$}
            \STATE \textbf{break}
        \ENDIF
        % Old version with Reduced-Mean subroutine integrated:
        % \STATE Set $\ell \gets 16\lceil\log(2/\hat\delta)/\hat\varepsilon^2 \rceil$
        % \hfill \change{// Lines 11--15 are the \textsc{Reduced-Mean} estimator in \citet{fahrbach2019submodular}}
        % \STATE Sample $X_1,X_2,\dots,X_\ell \sim \mathcal{D}_t$
        % \STATE Set $\overline{\mu} \leftarrow \frac{1}{\ell} \sum_{j=1}^\ell X_j$
        % \IF{$\overline{\mu} \le 1 - 1.5\hat\varepsilon$}
        %   \STATE \textbf{break}
        % \ENDIF
      \ENDFOR
      \STATE Sample $T \sim \mathcal{U}(A, \min\set{t, k - \abs{S}})$
      \STATE \change{Shuffle $T$ uniformly at random to get $(x_1, x_2,\dots,x_{|T|})$}
       \hfill \change{// Post-filtering}
      \STATE \change{Update $S' \gets S' \cup \{x_i \in T : \Delta(x_i, S \cup \{x_1, \dots, x_{i-1}\}) \ge \tau\}$} \label{step:post_filter}
      \STATE Update $S \leftarrow S \cup T$
      \IF{$\abs{S} = k$}
        \STATE \textbf{break}
      \ENDIF
    \ENDFOR
    % Old version that shuffled S and the looked at prefix gains.
    %\STATE \change{Shuffle $S$ uniformly at random to get $(x_1, x_2,\dots,x_{|S|})$}
    %\vspace{-0.40cm}
    %\STATE \change{Set $S' \gets \{x_i \in S : \Delta(x_i, \{x_1,\dots,x_{i-1}\}) \ge \tau \}$} \label{step:post_filter}
    \STATE \textbf{return} $(\change{S',}\ S, A)$
  \end{algorithmic}
\end{algorithm} 

\begin{restatable}{lemma}{NonmonotoneThresholdSamplingAlg}
\label{lem:adaptive_sampling}
\change{Let $Z$ be the event that all calls to \EstimateMean give correct outputs (i.e., the reported property in \Cref{lem:estimator} holds).}
\change{For any nonnegative submodular function $f$,
\NMTS} outputs \change{sets $S', S \subseteq N$ with $|S'| \le |S| \le k$} in
$O(\log(n/\delta)/\varepsilon)$
adaptive rounds such that the following \change{hold conditioned on $Z$}:
  \begin{enumerate}
    \item \change{The algorithm makes} $O(n/\varepsilon)$ oracle queries in expectation.
    \item \change{With probability at least $1 - \delta/2$,
         if $|S| < k$, the number of remaining candidates is $|A| < ck$.}
    \item \change{The filtered set $S'$ has expected size
          $\Exp{}{|S'| \mid Z} \ge (1 - \hat\varepsilon) \Exp{}{|S| \mid Z}$.}
  \end{enumerate}
  Further, event $Z$ happens with probability at least $1 - \delta/2$.
  Finally, the following properties hold \change{unconditionally}:
  \begin{enumerate}
    \setcounter{enumi}{3}
    \item \change{$f(S') \ge \tau \cdot |S'|$}
    \item \change{$\Pr(x \in S') \le 1/c$}
  \end{enumerate}
\end{restatable}

\begin{algorithm}[H]
  \caption{\EstimateMean}
  \label{alg:estimate_mean}
  \textbf{Input:} Bernoulli distribution $\mathcal{D}$,
    error $\varepsilon$, failure probability $\delta$
  \begin{algorithmic}[1]
    \STATE Set number of samples
      $m \leftarrow 16 \ceil{\log(2/\delta)/\varepsilon^2}$
    \STATE Sample $X_1, X_2,\dots, X_m \sim \mathcal{D}$
    \STATE Set $\overline{\mu} \leftarrow \frac{1}{m} \sum_{i=1}^m X_i$
    \IF{$\overline{\mu} \le 1 - 1.5\varepsilon$}
      \STATE \textbf{return} \texttt{true}
    \ENDIF
    \STATE \textbf{return} \texttt{false}
  \end{algorithmic}
\end{algorithm}

\begin{lemma}{\cite{fahrbach2019submodular}}
\label{lem:estimator}
For any Bernoulli distribution $\mathcal{D}$,
\EstimateMean uses $O(\log(\delta^{-1})/\varepsilon^{2})$
samples to \change{report one of the following properties,
which is correct with probability at least $1-\delta$:}
\begin{enumerate}
  \item If the output is \textnormal{\texttt{true}}, then the mean of $\mathcal{D}$ is $\mu \le 1 - \varepsilon$.
  \item If the output is \textnormal{\texttt{false}}, then the mean of $\mathcal{D}$ is $\mu \ge 1 - 2\varepsilon$.
\end{enumerate}
\end{lemma}

\subsection{Unconstrained Submodular Maximization}

The second subroutine in our non-monotone maximization algorithm is a
constant-approximation algorithm for unconstrained submodular maximization
that runs in a constant number of rounds depending on $\varepsilon$.
While the focus of this paper is submodular maximization subject to a
cardinality constraint, 
we show how calling \UnconstrainedMaximization
on a new ground set $A \subseteq N$ of size $O(k)$ can be used with
\citet{buchbinder2014submodular} to achieve a constant-approximation
for the constrained maximization problem.

\begin{lemma}{\cite{feige2011maximizing}}
\label{lem:unconstrained-maximization}
For any nonnegative submodular function $f$,
denote the solution to the unconstrained maximization problem by
$\OPT = \max_{S \subseteq N} f(S)$.
If $R$ is a uniformly random subset of~$S$, then
$\E[f(R)] \ge (1/4)\OPT$.
\end{lemma}

The guarantees for the \UnconstrainedMaximization algorithm
in \Cref{lem:unconstrained-maximization-alg}
are standard consequences of \Cref{lem:unconstrained-maximization}.

\begin{algorithm}
  \caption{\UnconstrainedMaximization}
  \label{alg:unconstrained-maximization}
  \textbf{Input:} oracle for $f : 2^N \rightarrow \R_{\ge 0}$,
  ground subset ${A \subseteq N}$,
  error $\varepsilon$, failure probability $\delta$
  \begin{algorithmic}[1]
    \STATE Set iteration bound $t \leftarrow \ceil{\log(1/\delta)/\log(1+(4/3)\varepsilon)}$
    \FOR{$i=1$ to $t$ in parallel}
      \STATE Let $R_i$ be a uniformly random subset of $A$
    \ENDFOR
    \STATE Set $S \leftarrow \argmax_{X \in \{R_1,\dots,R_t\}} f(X)$
    \STATE \textbf{return} $S$
  \end{algorithmic}
\end{algorithm} 

\begin{restatable}[]{lemma}{unconstrainedMaximizationAlg}
\label{lem:unconstrained-maximization-alg}
For any nonnegative submodular function~$f$ and
subset $A \subseteq N$,
\change{\Cref{alg:unconstrained-maximization}} outputs a set $S \subseteq A$
in one adaptive round
using $O(\log(1/\delta)/\varepsilon)$ oracle queries
such that with probability at least $1-\delta$ we have
$f(S) \ge (1/4 - \varepsilon)\OPT_A$, where
$\OPT_A = \max_{T \subseteq A} f(T)$.
\end{restatable}

An essentially optimal algorithm for unconstrained submodular
maximization was recently given in \citet{chen2018unconstrained},
which allows us to 
slightly improve the approximation
factor of our non-monotone maximization algorithm.

\begin{theorem}{\cite{chen2018unconstrained}}
\label{thm:new-unconstrained}
There is an algorithm that
achieves a $(1/2 - \varepsilon)$-approximation for unconstrained
submodular maximization
using $O(\log(1/\varepsilon) / \varepsilon)$ adaptive rounds
and $O(n \log^{3}(1/\varepsilon) / \varepsilon^4)$ evaluation oracle queries.
\end{theorem}

\section{Non-monotone Submodular Maximization}
\label{sec:nonmonotone}

In this section we show how to combine \change{\NMTS} and
\UnconstrainedMaximization to give the first constant-factor approximation
algorithm for non-monotone submodular maximization subject to a cardinality
constraint~$k$ that uses $O(\log(n))$ adaptive rounds.
Moreover, this algorithm makes $O(n \log(k))$ expected oracle queries.
While the approximation factor is only $0.039$,
we demonstrate that \change{threshold sampling} can readily be extended to
non-monotone settings without increasing adaptivity.

We start by describing \NonmonotoneMaximization
and the analysis of its approximation factor at a high level.
One inspiration for this algorithm is \Cref{lem:empty-lower-bound},
which allows us to lower bound the expected value of the returned set
$\Exp{}{f(R)}$ by $\OPT$ as long as
every element has at most a constant probability less than 1
of being in the output.
With this property in mind, \NonmonotoneMaximization
starts by trying different thresholds in parallel, one of which
is close to $c_1 \OPT/k$.
For each threshold, it runs \change{\NMTS and breaks}
if the number of candidates in $A$ falls below~$c_3 k$.
If $c_3 > 1$, this guarantees each element appears \change{in~$S'$} with
probability at most $1/c_3$.
In the event that \change{\NMTS} breaks because $|A| < c_3 k$,
it then runs unconstrained submodular maximization on $A$
and downsamples the solution \change{to have} cardinality at most~$k$.
In the end, the algorithm returns the set with the maximum value over all thresholds.
Our analysis shows how we optimize the constants $c_1$ and $c_3$
to balance the expected trade-offs between the two events
and give the best approximation factor.
%We present the algorithm and its guarantees below.

\begin{algorithm}
  \caption{\NonmonotoneMaximization}
  \label{alg:nonmonotone}
  \textbf{Input:} evaluation oracle for $f : 2^N \rightarrow \R$, constraint $k$, error $\varepsilon$, failure probability $\delta$
  \begin{algorithmic}[1]
    \STATE Set smaller error $\hat{\varepsilon} \leftarrow \varepsilon/6$ \label{alg-line:set-smaller-epsilon}
    \STATE Set $\Delta^* \leftarrow \max\{f(x) : x \in N\}$,
    $r \leftarrow \ceil{2\log(k)/\hat\varepsilon}$
    \STATE Set smaller failure probability $\hat{\delta} \leftarrow \delta/(2(r+1))$
    \STATE Set optimized constants
        $c_1 \leftarrow 1/7, c_3 \leftarrow 3$
    \STATE Initialize $R \leftarrow \emptyset$
    \FOR{$i=0$ to $r$ in parallel}
      \STATE Set $\tau \leftarrow c_1 (1+\hat\varepsilon)^i \Delta^*/k$ \label{alg-line:thresold-setting}
      \STATE \change{Set $(S', S, A) \leftarrow \textsc{Non-Monotone-}\textsc{Threshold-Sampling}(f, k, \tau, c_3, \hat\varepsilon, \hat{\delta})$} \label{alg-line:call-non-monotone-threshold-sampling}
      \STATE Initialize $U \leftarrow \emptyset, U' \leftarrow \emptyset$,
             $U'' \leftarrow \emptyset$
      \IF{$|A| < c_3 k$}
        \STATE Set $U \leftarrow \UnconstrainedMaximization(f, A, \hat\varepsilon, \hat\delta)$
        \IF{$|U| > k$}
          \STATE Sample $D \sim \mathcal{U}(U, k)$
          \STATE Update $U' \leftarrow D$
        \ELSE
          \STATE Update $U' \leftarrow U$
        \ENDIF
      \ENDIF
      \STATE Permute the elements of $U'$ uniformly at random
      \STATE Set $U'' \leftarrow$ highest-valued prefix of the permutation
      \STATE Update $R \leftarrow \argmax_{X \in \{R, \change{S',} U''\}} f(X)$
    \ENDFOR
    \STATE \textbf{return} $R$
  \end{algorithmic}
\end{algorithm}

\begin{restatable}[]{theorem}{nonmonotoneMaximization}
\label{thm:nonmonotone-maximization}
For any nonnegative submodular function $f$,
\NonmonotoneMaximization outputs a set $R \subseteq N$ with
$|R| \le k$ in $O(\log(n/\delta)/\varepsilon)$ adaptive rounds
\change{and with $O(n \log(k)/\varepsilon^2 + \delta n^2)$ oracle queries in expectation such that}
$\E[f(S)] \ge \change{0.026(1-\varepsilon)(1-\delta)}\OPT$.
\change{Setting $\delta < 1/n$ yields a good trade-off between the number of adaptive
rounds and oracle calls.}
\end{restatable}

Since the quality of our approximation relies on the approximation factor
of a low-adaptivity algorithm for unconstrained submodular maximization,
we can use \change{\citet{chen2018unconstrained}} instead of 
\UnconstrainedMaximization to 
improve our approximation without a loss in adaptivity or query
complexity.

% \begin{restatable}[]{theorem}{NewNonmonotoneMaximization}
% \label{thm:new-nonmonotone-maximization}
% There is an algorithm
% for nonnegative submodular maximization subject to a cardinality constraint $k$
% that achieves a $(0.039-\varepsilon)$-approximation in expectation
% using $O(\log(n)/\varepsilon)$ adaptive rounds and
% $O(n \log(k)/\varepsilon^2)$ expected queries to the evaluation oracle.
% \end{restatable}

\begin{restatable}[]{theorem}{NewNonmonotoneMaximization}
\label{thm:new-nonmonotone-maximization}
% Old version:
%There is an algorithm
%for nonnegative submodular maximization subject to a cardinality constraint $k$
%that achieves a $(0.039-\varepsilon)$-approximation in expectation
%using $O(\log(n)/\varepsilon)$ adaptive rounds and
%$O(n \log(k)/\varepsilon^2)$ expected queries to the evaluation oracle.
\change{For any nonnegative submodular function $f$,
there is an algorithm that outputs a set $R \subseteq N$ with
$|R| \le k$ in $O(\log(n/\delta)/\varepsilon)$ adaptive rounds
and with $O(n \log(k)\log^3(1/\varepsilon)/\varepsilon^5 + \delta n^2)$ expected oracle queries such that
$\E[f(S)] \ge \change{0.039 (1-\varepsilon)(1-\delta)} \OPT$.
Setting $\delta < 1/n$ yields a good trade-off between the number of adaptive
rounds and oracle calls.}
\end{restatable}

%%%%%%%%%%%%%%%%%%% Analysis of non-monotone %%%%%%%%%%%%%%%%%%%%
\subsection{Prerequisite Notation and Lemmas}
We start by defining notation that is useful for analyzing
\change{\NMTS}.
Let $T_1, T_2, \dots, T_r$ be the sequences of randomly generated sets
used to build the output set $S$.
Similarly, let the corresponding sequences of partial solutions be
$S_i = \bigcup_{j=1}^i T_j$ and candidate sets be
$A_0, A_1, \dots, A_r$.
To analyze the approximation factor of \NonmonotoneMaximization,
we consider a threshold $\tau$ sufficiently close to
$\tau^* = \OPT/k$ and then analyze the resulting sets \change{$S'$}, $U$, $U'$, and $U''$.
Lastly, we use $\ALG$ as an alias for the final output set $R$.

Next, we present several lemmas that are helpful for analyzing the
approximation factor.
%The following lemma is an equation in the proof of \Cref{lem:adaptive_sampling}.
%We use this lemma to show that the elements in any partial solution $S_i$
%have an average marginal gain exceeding the input threshold.
\morteza{The following lemma allows us to show that 
%(1) every element has at least a constant probability of not appearing in the output set, and (2) that 
the quality of a solution of size greater than $k$
degrades at worst by its downsampling rate.
%The first property is motivated by \Cref{lem:empty-lower-bound} and allows us
%to achieve a lower bound in terms of $\OPT$ in \Cref{lem:approx-lemma-2}.
This property is useful for analyzing Line~13 of the
\NonmonotoneMaximization algorithm.}

\begin{restatable}[]{lemma}{downsample}
\label{lem:downsample}
For any subset $S \subseteq N$ and $0 \le k \le |S|$,
if $T \sim \mathcal{U}(S, k)$,
then $\E[f(T)] \ge \morteza{\frac{k}{|S|}} \cdot f(S)$.
\end{restatable}

%\begin{proof}
%Let $s = |S|$ and
%fix an ordering $x_1, x_2,\dots, x_{s}$ on the elements in $S$.
%Expanding the expected value $\E[f(T)]$ and using submodularity,
%it follows that
%\begin{align*}
%  \Exp{}{f(T)} &= \frac{1}{\binom{s}{k}} \sum_{R \in \mathcal{U}(S, k)}
%    \sum_{x_i \in R} \Delta\parens*{x_i, \set*{x_1, \dots, x_{i-1}} \cap R}\\
%  &\ge
%  \frac{1}{\binom{s}{k}} \sum_{R \in \mathcal{U}(S, k)}
%    \sum_{x_i \in R} \Delta\parens*{x_i, \set*{x_1, \dots, x_{i-1}}}\\
%  &= \frac{1}{\binom{s}{k}} \sum_{i=1}^s  \binom{s-1}{k-1}
%    \Delta\parens*{x_i, \set*{x_1,\dots,x_{i-1}}}\\
%  &= \frac{k}{s} \cdot f(S). \tag*{\qedhere}
%\end{align*}
%\end{proof}

\change{We defer the proof of \Cref{lem:downsample} to~\Cref{app:nonmonotone}.}

\subsection{Analysis of the Approximation Factor}
\label{subsec:analysis}

The main idea behind our analysis is to capture two different behaviors 
of \NonmonotoneMaximization
and balance the worst of the two outcomes by optimizing constants.

\begin{definition}
\change{For a given threshold,}
let $A_{< c_3 k}$ denote the event that
\change{\NMTS} breaks because $|A| < c_3 k$.
Similarly, let $A_{\ge c_3 k}$ denote the complementary event.
\end{definition}

The following two key lemmas lower bound the expected solution in terms of
$\OPT$ and $\Prob{}{A_{< c_3 k} \mid Z}$,
\change{where $Z$ is the event that all \textsc{Reduced-Mean} estimates (defined in \Cref{lem:adaptive_sampling}),
for all candidate thresholds and their internal rounds, output correct properties.}
The goal is to average these inequalities
so that the conditional probability terms disappear,
giving us with a lower bound \change{that is} only in terms of $\OPT$.

\change{
\begin{lemma}\label{lem:approx-lemma-1}
For any $\tau$ such that $\tau \le c_1 \tau^* \le \tau(1+\hat{\varepsilon})$,
we have
\[
    \Exp{}{\ALG}
    \ge
    \Exp{}{f(S') \mid Z} \Prob{}{Z}
    \ge
    (1 - \varepsilon) (1 - \morteza{\delta}) c_1 \OPT \cdot \Prob{}{A_{\ge c_3 k} \mid Z}.
\]
\morteza{Here $\varepsilon$ and $\delta$ refer to the input parameters of Algorithm~\ref{alg:nonmonotone}, and $\hat{\varepsilon}$ is set in Line~\ref{alg-line:set-smaller-epsilon} of that algorithm.}
\end{lemma}

\begin{proof}
\morteza{
The way we set the threshold value $\tau$ in Line~\ref{alg-line:thresold-setting} ensures that
there exists a value of $\tau$ that satisfies the requirements of the lemma statement, i.e., $\tau \approx c_1 \tau^*$. Throughout the rest of the proof, when we talk about sets $S$ and $S'$, we refer to the output of \textsc{Non-Monotone-}\textsc{Threshold-Sampling} in Line~\ref{alg-line:call-non-monotone-threshold-sampling} for this particular value of $\tau$.}
First observe that by \Cref{lem:adaptive_sampling} Property 4, we have
\begin{align*}
    \Exp{}{\ALG}
    &\ge
    \Exp{}{f(S')} \\
    &\ge
    \Exp{}{f(S') \mid Z} \Prob{}{Z} \\
    &\ge
    \tau \cdot \Exp{}{|S'| \mid Z} \Prob{}{Z}.
\end{align*}
Using Properties~3 and~2 of~\Cref{lem:adaptive_sampling}
and applying the law of total expectation
since \morteza{both $|S|$ and $|S'|$ are nonnegative random variables},
\begin{align*}
    \Exp{}{|S'| \mid Z} \Prob{}{Z}
    &\ge
    (1 - \hat\varepsilon) \cdot \Exp{}{|S| \mid Z} \Prob{}{Z} \\
    &\ge
    (1 - \hat\varepsilon) \cdot \Exp{}{|S| \mid Z, A_{\ge c_3 k}} \Prob{}{Z, A_{\ge c_3 k}} \\
    &\ge
    (1 - \hat\varepsilon) k (1-\hat\delta/2) \Prob{}{Z, A_{\ge c_3 k}},
\end{align*}
since \morteza{conditioned on $Z$}, with probability at least $1 - \hat\delta/2$, we have
$|S| = k$ if $|A| \ge c_3 k$.
Rearranging the conditional probabilities, it follows that
\begin{align*}
    (1 - \hat\varepsilon) k (1-\hat\delta/2) \Prob{}{Z, A_{\ge c_3 k}}
    &=
    (1 - \hat\varepsilon) k (1-\hat\delta/2) \Prob{}{A_{\ge c_3 k} \mid Z} \Prob{}{Z} \\
    &\ge
    (1 - \hat\varepsilon) k (1-\hat\delta/2) \Prob{}{A_{\ge c_3 k} \mid Z} (1 - \morteza{\hat\delta}/2) \\
    &\ge
    (1 - \hat\varepsilon) k (1-\morteza{\hat\delta}) \Prob{}{A_{\ge c_3 k} \mid Z}.
\end{align*}
Putting everything together including our choice of $\tau$ gives us
\begin{align*}
    \Exp{}{\ALG}
    &\ge
    \tau \cdot \Exp{}{|S'| \mid Z} \Prob{}{Z} \\
    &\ge
    (1 - \hat\varepsilon) (1-\morteza{\hat\delta}) k \tau \Prob{}{A_{\ge c_3 k} \mid Z} \\
    &\ge
    (1 - \varepsilon) (1 - \delta) c_1 \OPT \cdot \Prob{}{A_{\ge c_3 k} \mid Z}. \qedhere
\end{align*}
% OLD VERSION:
% First observe that $\ExpCond{}{|S|}{A_{\ge c_3 k}} = k$.
% Properties 2 and 3 of~\Cref{lem:adaptive_sampling}
% combined with the law of total expectation imply that
% \begin{align*}
%   \Exp{}{f(S')}
%   &\ge
%   \tau \cdot \Exp{}{|S'|} \\
%   &\ge
%   (1-\hat{\varepsilon}) \tau \cdot \Exp{}{|S|} \\
%   &\ge
%   (1-\hat{\varepsilon}) \tau \cdot \ExpCond{}{|S|}{A_{\ge c_3 k}} \cdot \Prob{}{A_{\ge c_3 k}} \\
%   &=
%   (1-\hat{\varepsilon}) \tau k \Prob{}{A_{\ge c_3 k}} \\
%   &\ge
%   (1-\varepsilon) c_1 \Prob{}{A_{\ge c_3 k}} \cdot \OPT.
% \end{align*}
% The result follows from the fact $\Exp{}{\ALG} \ge \Exp{}{f(S')}$.
\end{proof}
}

The core of our analysis proves the following lower bound,
which intricately uses the conditional expectation of nonnegative random variables.

\begin{lemma}\label{lem:approx-lemma-2}
Let $\UnconstrainedApprox$ denote the approximation factor for an unconstrained
submodular maximization algorithm.
For any threshold $\tau$ such that
$\tau \le c_1 \tau^* \le \tau(1+\hat\varepsilon)$, we have
\[
  \Exp{}{\ALG} \ge
  \frac{(1 - \varepsilon) (1-\change{\delta}) \UnconstrainedApprox}{c_3}
    \bracks*{(1-c_1) \OPT \cdot \Prob{}{A_{< c_3 k} \change{\mid Z}} - \frac{\OPT}{c_3}
    - 2\Exp{}{f(\change{S'} \mid \change{Z})}}.
\]
\end{lemma}

\begin{proof} %[Proof Sketch.]
\morteza{Similar to \Cref{lem:approx-lemma-1}, we know that there exists a value of $\tau$ in the algorithm such that $\tau \le c_1 \tau^* \le \tau(1+\hat\varepsilon)$.
We focus on sets $S', S, A$ returned for this particular value of $\tau$.}
For any \change{triple} of subsets $\change{S'}, S, A \subseteq N$ returned by \change{$\NMTS$},
we can partition the optimal \change{set} $S^*$ into
$S_1^* = S^* \cap A$ and
$S_2^* = S^* \setminus A$.
Let $U_A$ be the output of \change{the} $\UnconstrainedMaximization$
\change{call on Line~11.}
By \Cref{lem:unconstrained-maximization-alg} \change{and a union bound}, we have
$f(U_A) \ge (\alpha - \hat\varepsilon) f(S_1^*)$
\change{with probability at least $1 - \delta/2$.}
Submodularity and the definition of $A$ imply that
$f(S_2^* \cup \change{S'}) \le f(\change{S'}) + k\tau$.
Let $\Gap(A,\change{S'}) = \max\{f(S_2^*) - f(S_2^* \cup \change{S'}), 0\}$.
By subadditivity and the previous inequalities,
\change{we have}
\begin{align}\label{eqn:sketch-1}
  f(S^*)
  &\le
  \change{f(S_1^*) + f(S_2^*)} \notag\\
  &=
  f(S_1^*) + f(S_2^*) - f(S_2^* \cup \change{S'}) + f(S_2^* \cup \change{S'}) \notag\\
  &\le
  (\alpha - \hat\varepsilon)^{-1} f(U_A) + \Gap(A, \change{S'}) + f(\change{S'}) + k\tau.
\end{align}
Using \Cref{eqn:sketch-1} and the assumption on $\tau$, \change{it follows that}
  \begin{align}\label{eqn:sketch-2}
  f(U_A)
  \ge
  (\alpha - \hat\varepsilon)
  \bracks*{(1 - c_1) \OPT - \Gap(A,\change{S'}) - f(\change{S'})}.
\end{align}

Next, our goal is to upper bound $\Gap(A,\change{S'})$ as a function of $S^*$
\change{to get} a bound that is independent of $A$.
Specifically, we prove in \change{\Cref{app:nonmonotone} that for any sets} $\change{S'},A \subseteq N$,
\change{we have}
\[
    \change{f(S_2^*) - f(S_2^* \cup \change{S'}) \le f(S^*) - f(S^* \cup \change{S'}).}
\]
This is a consequence of submodularity,
\change{and it implies $\Gap(A,S') \le \Gap(\emptyset, S')$.}
\change{Finally}, we have \change{the inequality}
\begin{align*}
    \Gap(\emptyset, \change{S'})
    &=
    \max\{f(S^*) - f(S^* \cup S'), 0\} \\
    &\le
    f(S^*) - f(S^* \cup \change{S'}) + f(\change{S'}).
\end{align*}
\change{Observe that if $f(S^*) - f(S^* \cup S') < 0$
then $\Gap(\emptyset, \change{S'}) = 0$,
but we have $f(S^*) - f(S^* \cup \change{S'}) + f(\change{S'}) \ge 0$
by subadditivity
since $f$ is nonnegative.}

Next, define a new submodular function $g : 2^N \rightarrow \R_{\ge 0}$
such that $g(S) = f(S^* \cup S)$.
Consider the set $\change{S'}$ that is returned by \change{\NMTS.
Each element appears in $\change{S'}$ with probability at most $1/c_3$ by
\Cref{lem:adaptive_sampling}.
Applying \Cref{lem:empty-lower-bound} to $g$ then gives us
$\Exp{}{f(S'\cup S^*) \mid Z} \ge (1-1/c_3) f(S^*)$.}
\morteza{We note that even when conditioning on $Z$, the probability of each element being in $S'$ is upper bounded by $1/c_3$ as elaborated in the proof of Property 5 of Lemma~\ref{lem:adaptive_sampling}. This makes \Cref{lem:empty-lower-bound} applicable.}
It follows that
\begin{align}\label{eqn:sketch-3}
  \Exp{}{\Gap(\emptyset,\change{S'}) \mid \change{Z}}
  &\le
  \Exp{}{f(S^*) - f(S^* \cup \change{S'}) + f(\change{S'}) \mid \change{Z}} \notag\\
  &\le
  (1/c_3) \OPT + \Exp{}{f(\change{S'}) \mid \change{Z}}.
\end{align}

\change{
Now we are prepared to give the lower bound for $\Exp{}{\ALG}$.
It follows from \Cref{lem:downsample} that
\begin{align}
\label{eqn:alg_lower_bound_case_2}
  \Exp{}{\ALG}
  \ge
  \Exp{}{f(U'')}
  \ge
  \Exp{}{f(U')}
  \ge
  \frac{1}{c_3}
  \Exp{}{f(U)}.
\end{align}
Since $f$ is nonnegative, the law of total expectation gives
\[
  \Exp{}{f(U)}
  \ge 
  \ExpCond{}{f(U)}{Z, A_{<c_3 k}} \cdot \Prob{}{Z, A_{<c_3 k}}.
\]
If events $Z$ and $A_{<c_3 k}$ both occur,
\UnconstrainedMaximization is called and we can
lower bound
$\ExpCond{}{f(U)}{Z, A_{< c_3 k}}$ using inequalities~\Cref{eqn:sketch-2} and~\Cref{eqn:sketch-3}.
Concretely, we have
\begin{align*}
  \ExpCond{}{f(U)}{Z, A_{< c_3 k}} \Prob{}{Z, A_{<c_3 k}}
  &\ge
   (1-\delta/2)(\alpha - \hat\varepsilon) \bracks*{(1-c_1)\OPT - \ExpCond{}{\Gap(A,S') + f(S')}{Z, A_{< c_3 k}}} \Prob{}{Z, A_{<c_3 k}}
  \\  
  &\ge
  (1-\delta/2)(\alpha - \hat\varepsilon) \bracks*{(1-c_1)\OPT - \ExpCond{}{\Gap(\emptyset,S') + f(S')}{Z, A_{< c_3 k}}} \Prob{}{Z, A_{<c_3 k}} 
  \\  
  &\ge
  (1-\delta/2)(\alpha - \hat\varepsilon) \bracks*{(1-c_1)\OPT  \cdot \Prob{}{Z, A_{<c_3 k}}  - \Exp{}{\Gap(\emptyset,S') + f(S') \mid Z} \Prob{}{Z}}
 \\
  &\ge
  (1-\delta/2)(\alpha - \hat\varepsilon) \bracks*{(1-c_1)\OPT  \cdot \Prob{}{A_{<c_3 k} \mid Z}  - \frac{\OPT}{c_3} - 2\Exp{}{f(S') \mid Z}{}} \Prob{}{Z}
  \\
  &\ge
  (1-\delta)(1-\varepsilon)
  \alpha
  \bracks*{(1-c_1) \OPT \cdot \Prob{}{A_{<c_3 k} \mid Z} - \frac{\OPT}{c_3} - 2\Exp{}{f(S') \mid Z}}.
\end{align*}
\morteza{We first note that the $1-\delta/2$ factor that appears in the first inequality accounts for the success probability of the call to \UnconstrainedMaximization.}
The second inequality holds because we showed that $\Gap(A,S) \leq \Gap(\emptyset, S)$ for any $A$.
The third inequality follows from
the law of total expectation since $Z$ is a superset of the event $Z$ and $A_{< c_3, k}$,
and observing that~$f$ and $\Gap(A,S)$ are nonnegative.
The fourth inequality applies~Equation~\eqref{eqn:sketch-3}.
The last inequality holds because $\hat{\varepsilon} = \varepsilon/6$ and
$\alpha \ge 1/4$ by \Cref{lem:unconstrained-maximization-alg},
\morteza{and $\Prob{}{Z} \ge 1 - \delta/2$ by Lemma~\ref{lem:adaptive_sampling}}.
Combining this with the lower bound for $\Exp{}{\ALG}$ in
Equation~\Cref{eqn:alg_lower_bound_case_2} completes the proof.
}
\end{proof}

Equipped with these two complementary lower bounds, we now prove our main
results.

\begin{proof}[Proof of \Cref{thm:nonmonotone-maximization}]
% First assume that all subroutines behave as desired with probability at least
% $1-\delta$ by our choice of~$\hat\delta$ and a union bound.
Since \NonmonotoneMaximization tries a $\tau$
such that $\tau \le c_1 \tau^* \le \tau(1+\hat\varepsilon)$,
\change{let the analysis follow for this particular threshold.}

We start with the proof of the approximation factor.
Suppose $\Exp{}{f(\change{S'}) \mid \change{Z}} > c_4 \OPT$ for a constant
$c_4 \ge 0$ that we later optimize.
This leads to a $(1-\delta)c_4$-approximation for $\OPT$.
Otherwise, \change{if} $\Exp{}{f(\change{S'}) \mid \change{Z}} \le c_4 \OPT$,
it follows from \Cref{lem:approx-lemma-2} that
\begin{align}\label{eqn:approx-2}
  \Exp{}{\ALG} &\ge \frac{(1 - \varepsilon) \change{(1-\delta)}\UnconstrainedApprox }{c_3} 
      \bracks*{(1-c_1) \Prob{}{A_{< c_3 k} \change{\mid Z}} - \frac{1}{c_3} - 2 c_4} \OPT.
\end{align}
Taking a weighted average of \Cref{lem:approx-lemma-1} and Equation~\Cref{eqn:approx-2}
\change{with coefficients $\beta \ge 0$ and $1$}
gives us
\begin{align*}
  \frac{\Exp{}{\ALG}}{\OPT} &\ge \frac{(1-\varepsilon)\change{(1-\delta)}\UnconstrainedApprox}{(1+\beta)c_3} \bracks*{\frac{\beta c_1 c_3}{\UnconstrainedApprox} \Prob{}{A_{\ge c_3 k}\change{\mid Z}}
    + (1-c_1) \Prob{}{A_{< c_3 k}\change{\mid Z}} - \frac{1}{c_3} - 2c_4}.
\end{align*}
To bound the approximation factor, we solve the optimization problem
\begin{align*}
  \max_{c_1, c_3, c_4, \beta} \min \set*{\change{(1-\delta)}c_4, \frac{(1-\varepsilon)\change{(1-\delta)}\UnconstrainedApprox}{(1+\beta)c_3}
            \parens*{1 - c_1 - \frac{1}{c_3} - 2c_4} }
\end{align*}
subject to the constraint $\beta c_1 c_3 /\UnconstrainedApprox = 1-c_1$
\change{in order to balance} the two complementary \change{probabilities}.

Now we optimize the constants in the algorithm.
The equality constraint implies that
$c_1 = (1 + \beta c_3 \UnconstrainedApprox^{-1})^{-1}$.
Next, we set the two expressions in the maximin problem to
be equal since one is increasing in $c_4$ and one is decreasing \change{in~$c_4$.}
\change{Canceling the $(1-\delta)$ factors for now,}
this implies that
\begin{align*}
  %&c_4 = \frac{(1-\varepsilon)\UnconstrainedApprox}{(1+\beta)c_3} \cdot \parens*{1 - c_1 - 1/c_3 - 2c_4} \\
  %&\implies
  c_4 = \frac{(1-\varepsilon) \change{\UnconstrainedApprox} \parens*{1-c_1 - c_3^{-1}}}{
    (1+\beta)c_3 + 2(1-\varepsilon)\change{\UnconstrainedApprox}}.
\end{align*}
Using the expressions above for $c_4$ and $c_1$, it follows that
\begin{align}\label{eqn:reduced-program}
  \frac{\Exp{}{\ALG}}{\OPT} \ge 
    \frac{(1-\varepsilon)\change{(1-\delta)\UnconstrainedApprox}\parens*{1- (1+\beta c_3 \alpha^{-1})^{-1} - c_3^{-1}}}{
      (1+\beta)c_3 + 2\change{\UnconstrainedApprox}}.
\end{align}
\Cref{lem:unconstrained-maximization-alg} \change{gives us} $\alpha = 1/4$.
Setting $c_3 = 3$, $\beta = 1/2$, it follows that
$c_1 = 1/7$.
\change{Therefore, putting everything together, we get}
an approximation factor of $\change{(1-\varepsilon)(1-\delta)11/420 \ge 0.026(1-\varepsilon)(1-\delta)}$
by \change{Equation}~\Cref{eqn:reduced-program}.

The proof of the adaptivity and query complexities follow from
\Cref{lem:adaptive_sampling} and \Cref{lem:unconstrained-maximization-alg}
since all $O(\log(k)/\varepsilon)$ thresholds are run in parallel.
This completes the analysis for the \NonmonotoneMaximization algorithm.
\end{proof}

\begin{proof}[Proof of \Cref{thm:new-nonmonotone-maximization}]
The proof is \change{identical} to the proof of \Cref{thm:nonmonotone-maximization}
except that \Cref{thm:new-unconstrained} \change{gives us} $\UnconstrainedApprox = 1/2$.
Setting $c_3=3.556,\beta=0.5664$, we have
$c_1 = 0.198989$, \change{which gives us an approximation factor of} $(1-\varepsilon)(1-\delta)0.0395$.
\change{The query complexity, however, increases because of the new subroutine for
unconstrained maximization in \Cref{thm:new-unconstrained}.}
\end{proof}

\section{Experiments}
\label{sec:experiments}

In this section, we evaluate \NonmonotoneMaximization on three real-world
applications introduced in~\citet{mirzasoleiman2016fast}.
We compare our algorithm with several benchmarks for non-monotone
submodular maximization and demonstrate that it consistently
finds competitive solutions using significantly fewer rounds and queries.
Our experiments build on those in \citet{balkanski2018non},
which plot function values at each round as the algorithms progress.
Additionally, we include plots of $\max_{|S| \le k} f(S)$ for
different constraints $k$ and plots of
the cumulative number of queries an algorithm has used after each round.
For algorithms that rely on a $(1\pm\varepsilon)$-approximation of $\OPT$,
we run all guesses in parallel and record statistics for the approximation
that maximizes the objective function.
We defer the implementation details to the supplementary manuscript.

Next, we briefly describe the benchmark algorithms.
The \Greedy algorithm builds a solution by choosing an
element with the maximum positive marginal gain in each round. This
requires $O(k)$ adaptive rounds and $O(nk)$ oracle queries, and it does not
guarantee a constant approximation.
The \Random algorithm randomly permutes the ground set and
returns the highest-valued prefix of elements.
It uses a constant number of rounds, makes~$O(k)$ queries, and 
also fails to give a constant approximation.
The \RandomLazyGreedy algorithm~\cite{buchbinder2016comparing} lazily builds a
solution by randomly selecting one of the $k$ elements with highest marginal
gain in each round.
This gives a $(1/e-\varepsilon)$-approximation in $O(k)$ adaptive rounds using
$O(n)$ queries.
The \Fantom algorithm~\cite{mirzasoleiman2016fast} is similar to
\Greedy and robust to intersecting matroid and knapsack constraints.
For a cardinality constraint, it gives a $(1/6 - \varepsilon)$-approximation
using $O(k)$ adaptive rounds and $O(nk)$ queries.
The \Blits algorithm~\cite{balkanski2018non} constructs a solution by
randomly choosing blocks of high-valued elements, giving a
$(1/(2e)-\varepsilon)$-approximation in $O(\log^2(n))$ rounds.
While \Blits is exponentially faster than the previous algorithms,
it requires $O(\OPT^2 n \log^2(n)\log(k))$ oracle queries.

\begin{figure*}
\centering

\begin{subfigure}{.247\linewidth}
  \centering
  \includegraphics[width=1.0\linewidth]{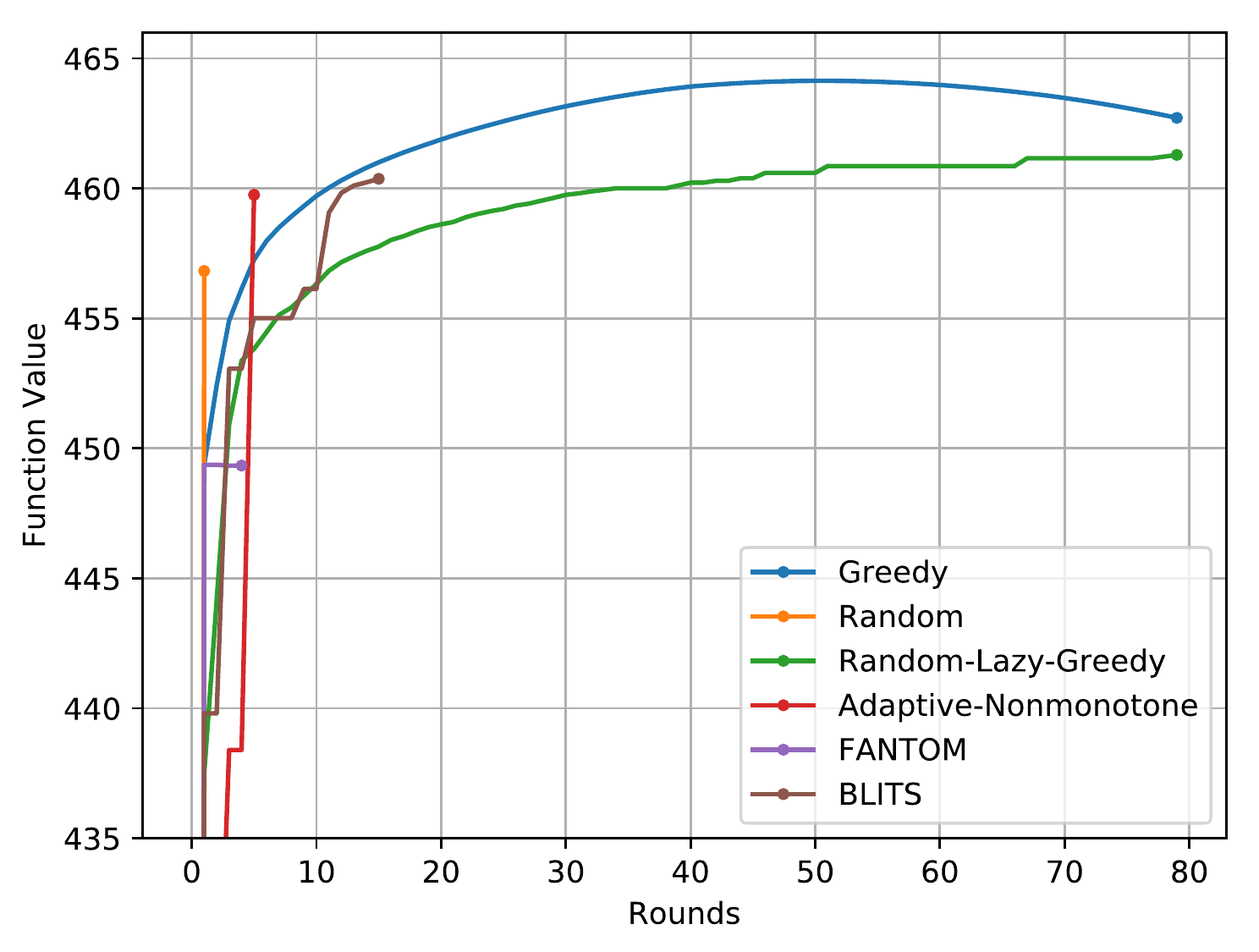}
  \caption{Image Summarization}
  \label{fig:images-rounds}
\end{subfigure}%
\begin{subfigure}{.247\linewidth}
  \centering
  \includegraphics[width=1.0\linewidth]{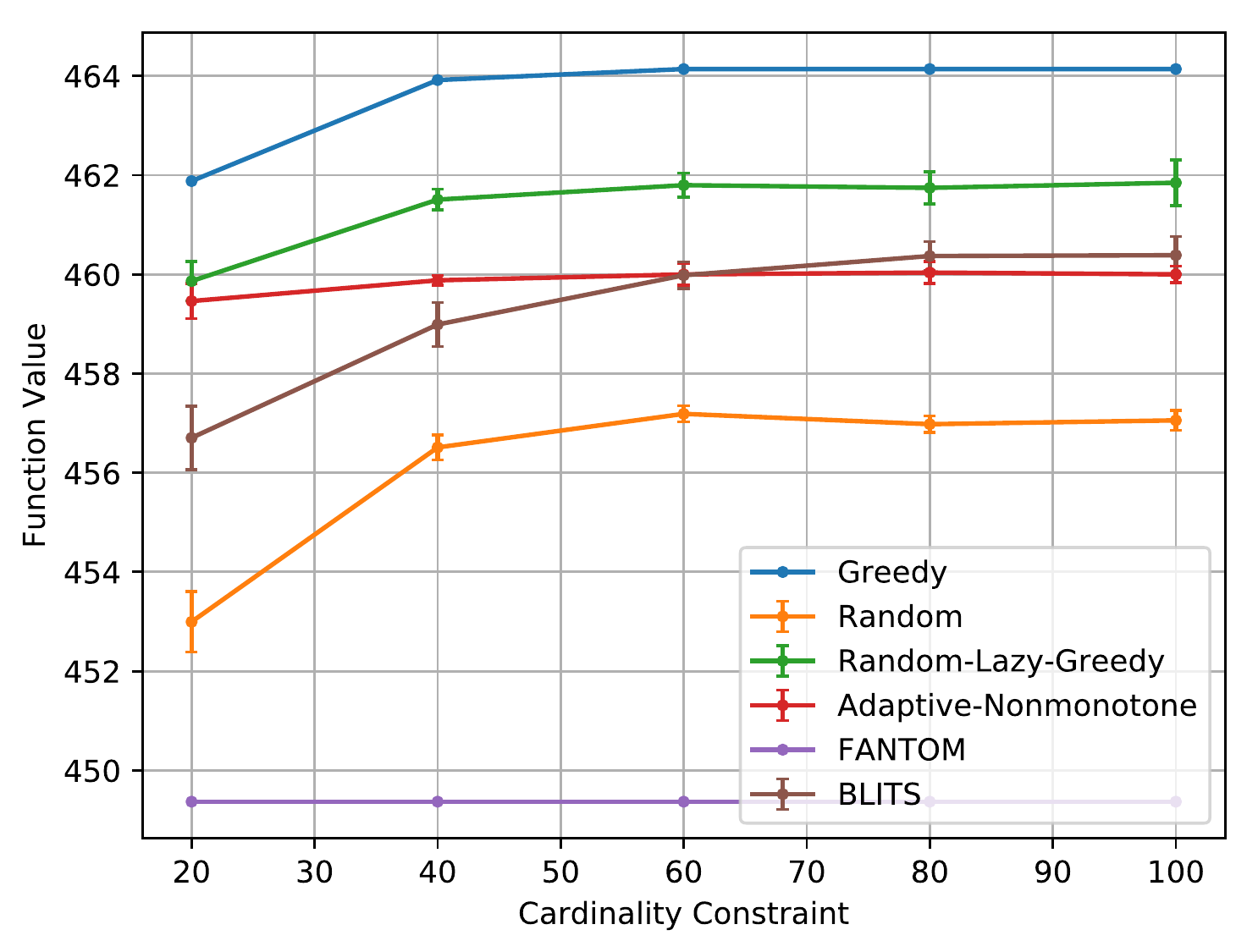}
  \caption{Image Summarization}
  \label{fig:images-constraints}
\end{subfigure}
\begin{subfigure}{.247\linewidth}
  \centering
  \includegraphics[width=1.0\linewidth,height=3.2cm]{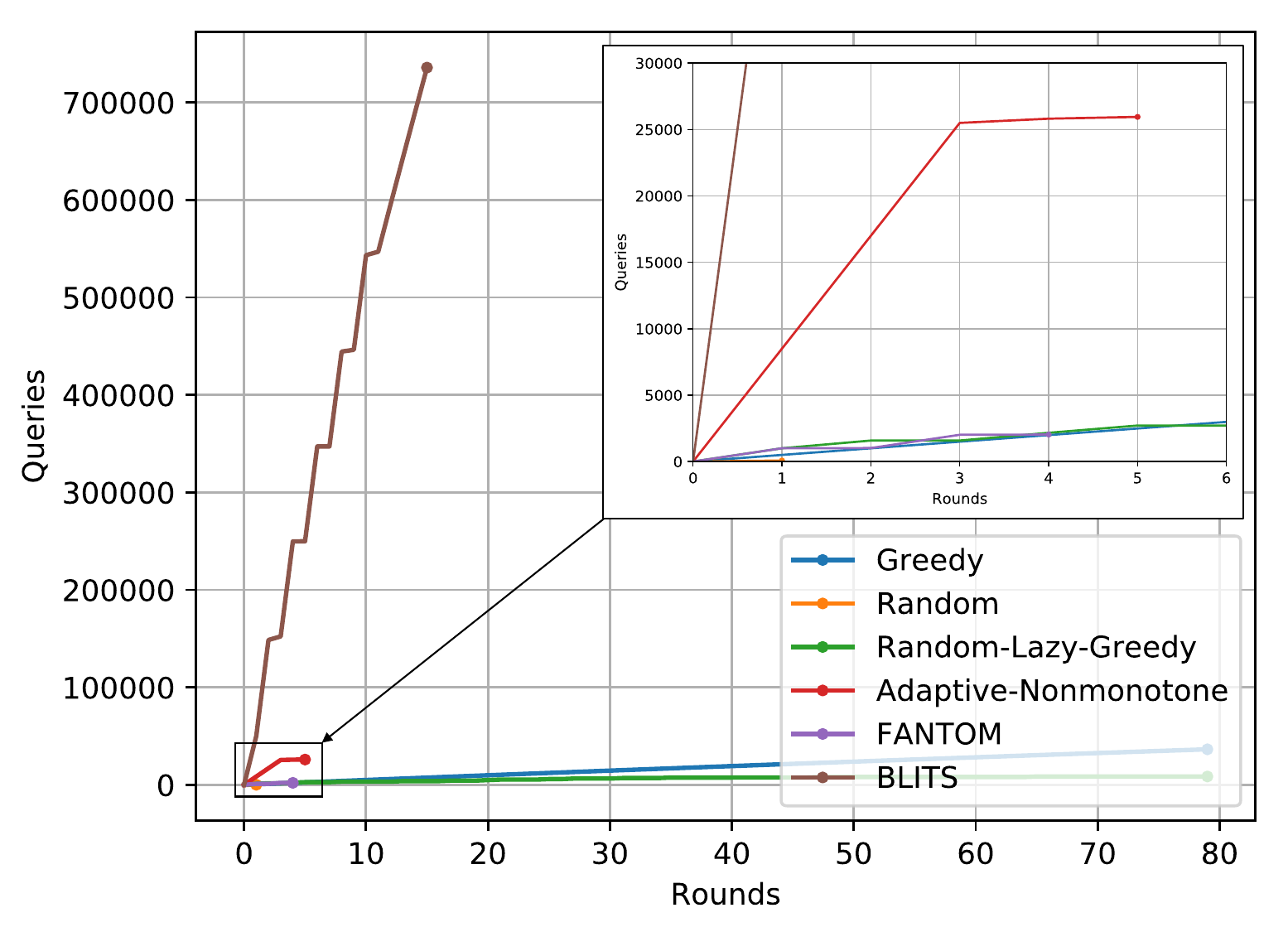}
\caption{Image Summarization}
\label{fig:images-queries}
\end{subfigure}
\begin{subfigure}{.247\linewidth}
  \centering
  \includegraphics[width=1.0\linewidth]{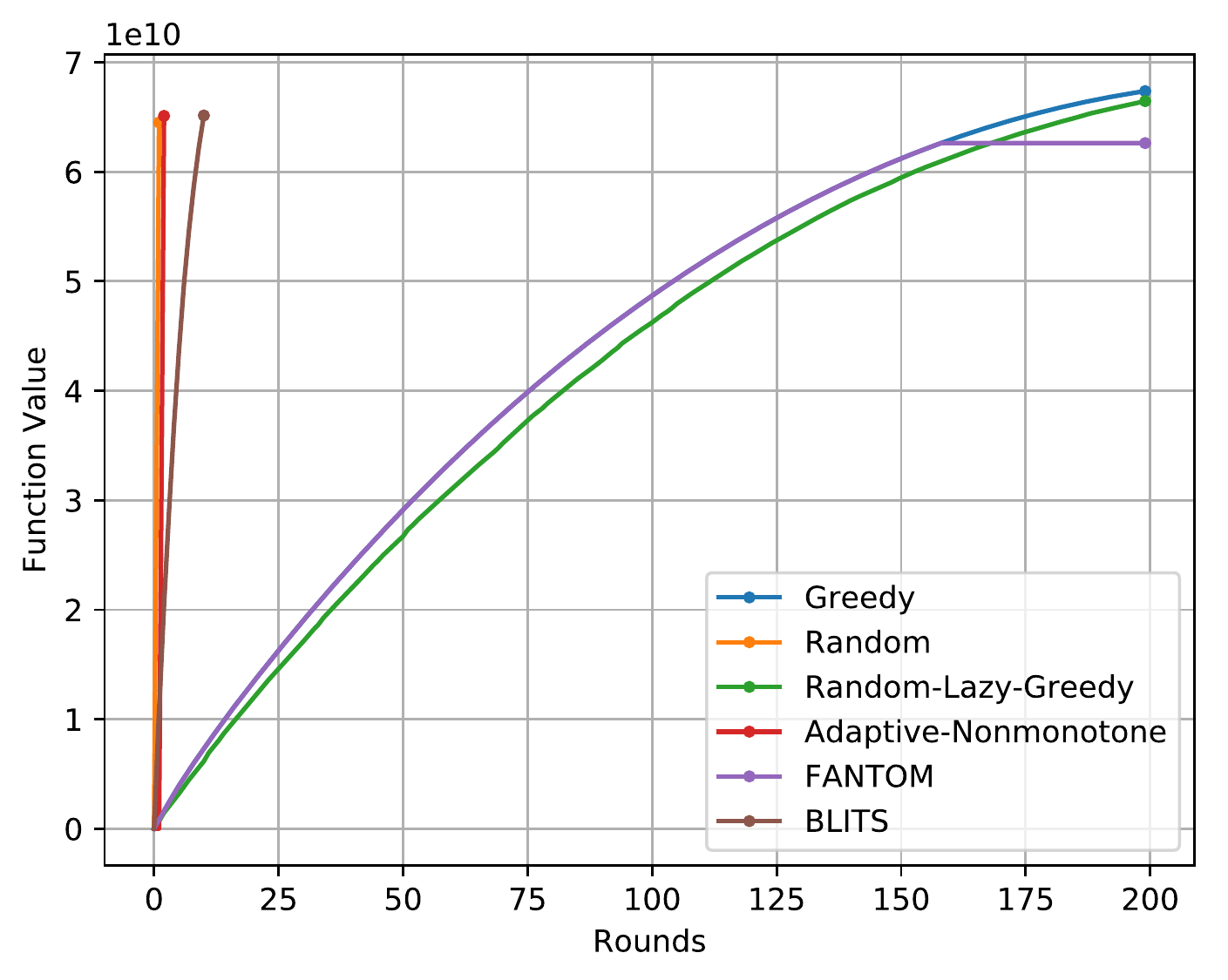}
  \caption{Movie Recommendation}
  \label{fig:movie-rounds}
\end{subfigure}%

% Next row

\begin{subfigure}{.247\linewidth}
  \centering
  \includegraphics[width=1.0\linewidth]{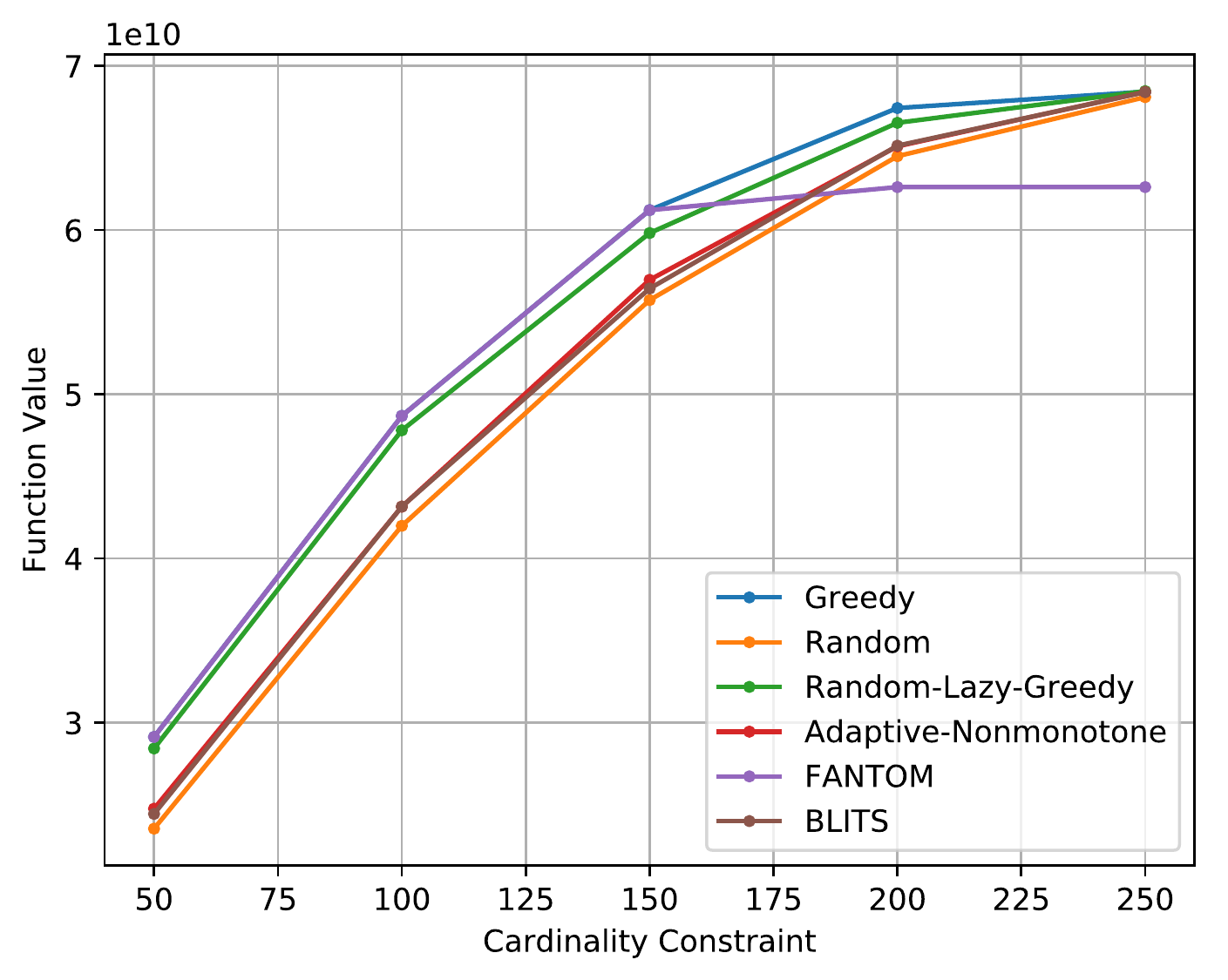}
\caption{Movie Recommendation}
\label{fig:movie-constraints}
\end{subfigure}
\begin{subfigure}{.247\linewidth}
  \centering
  \includegraphics[width=1.0\linewidth]{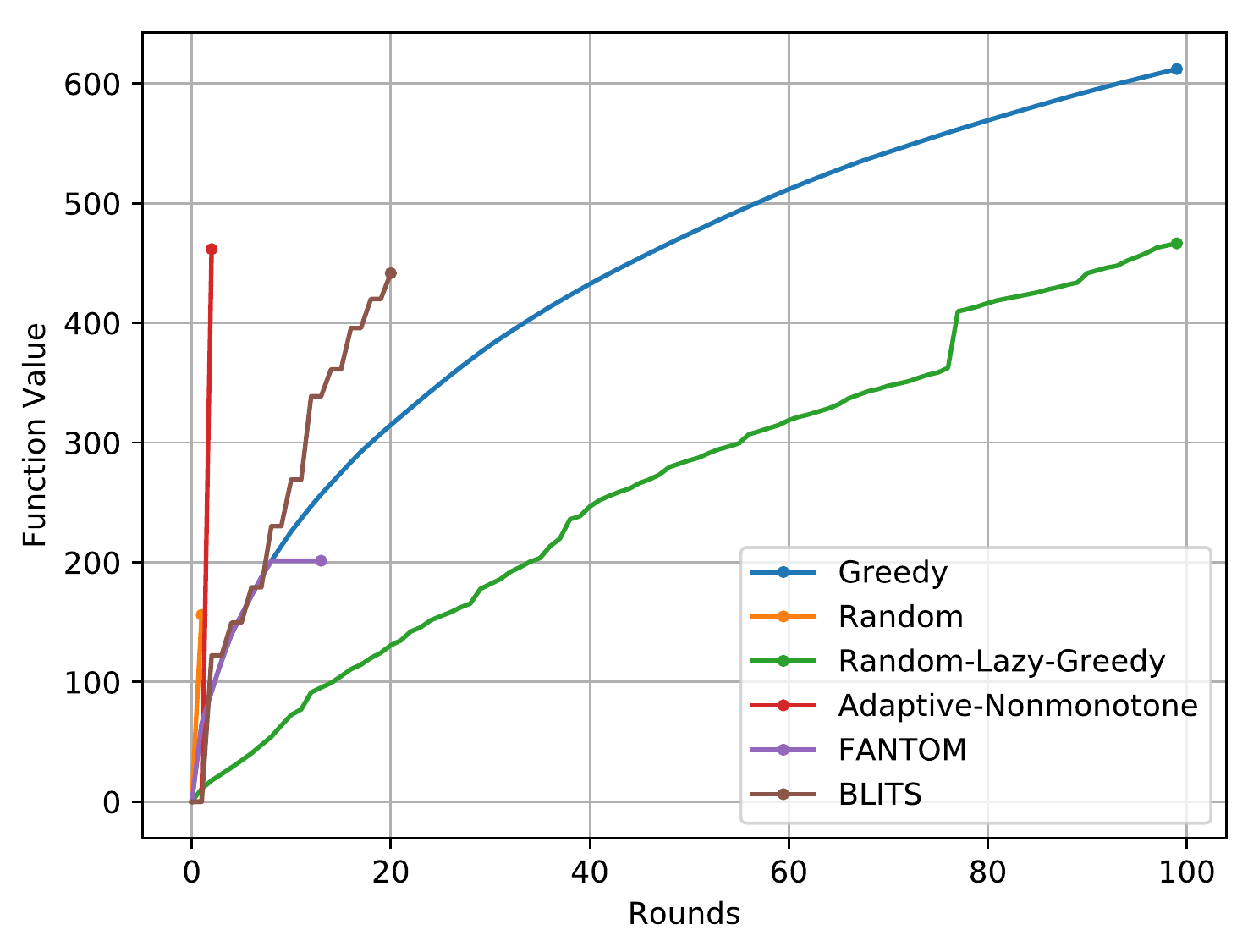}
  \caption{Revenue Maximization}
  \label{fig:youtube-rounds}
\end{subfigure}%
\begin{subfigure}{.247\linewidth}
  \centering
  \includegraphics[width=1.0\linewidth]{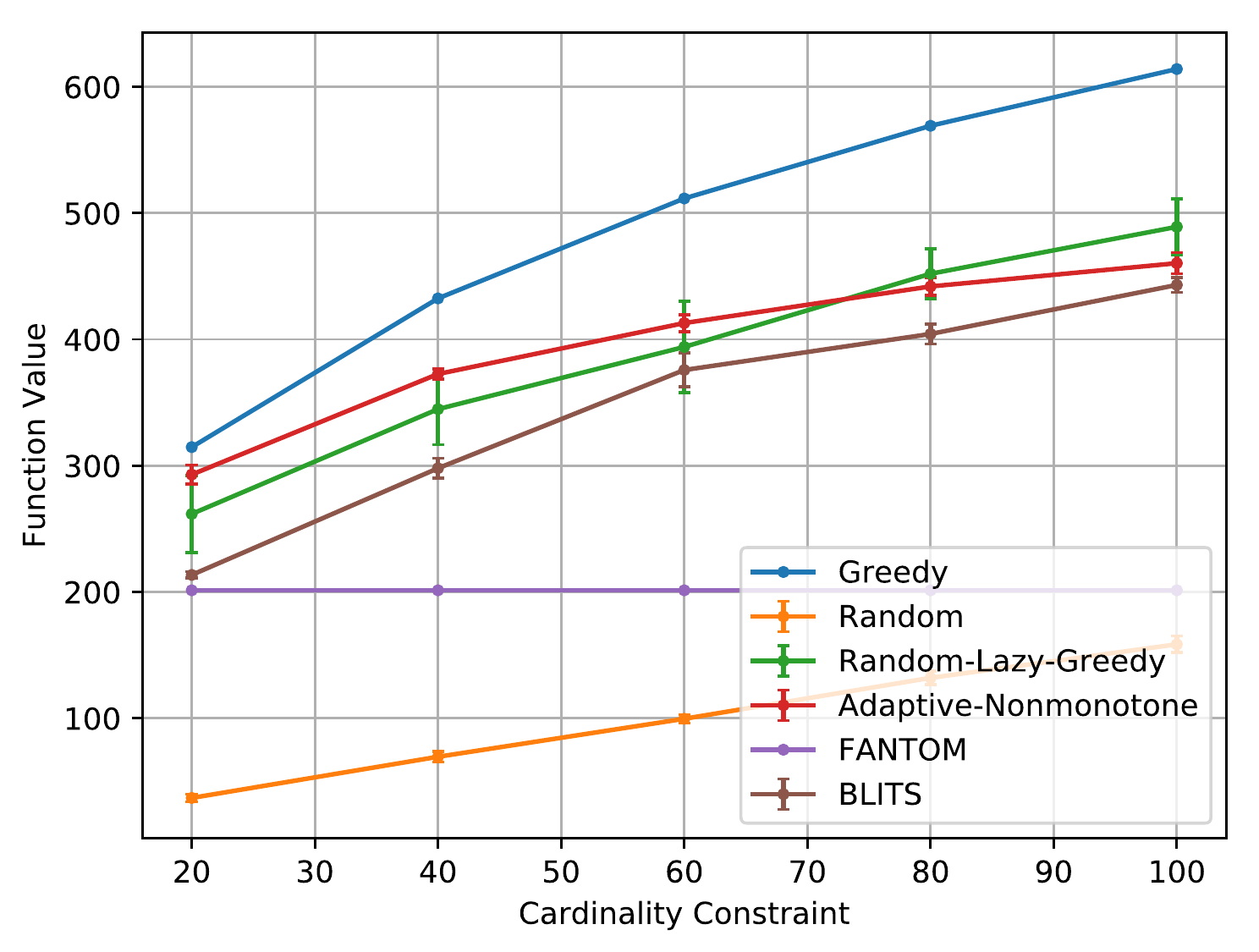}
  \caption{Revenue Maximization}
  \label{fig:youtube-constraints}
\end{subfigure}
\begin{subfigure}{.247\linewidth}
  \centering
  \includegraphics[width=1.0\linewidth,height=3.2cm]{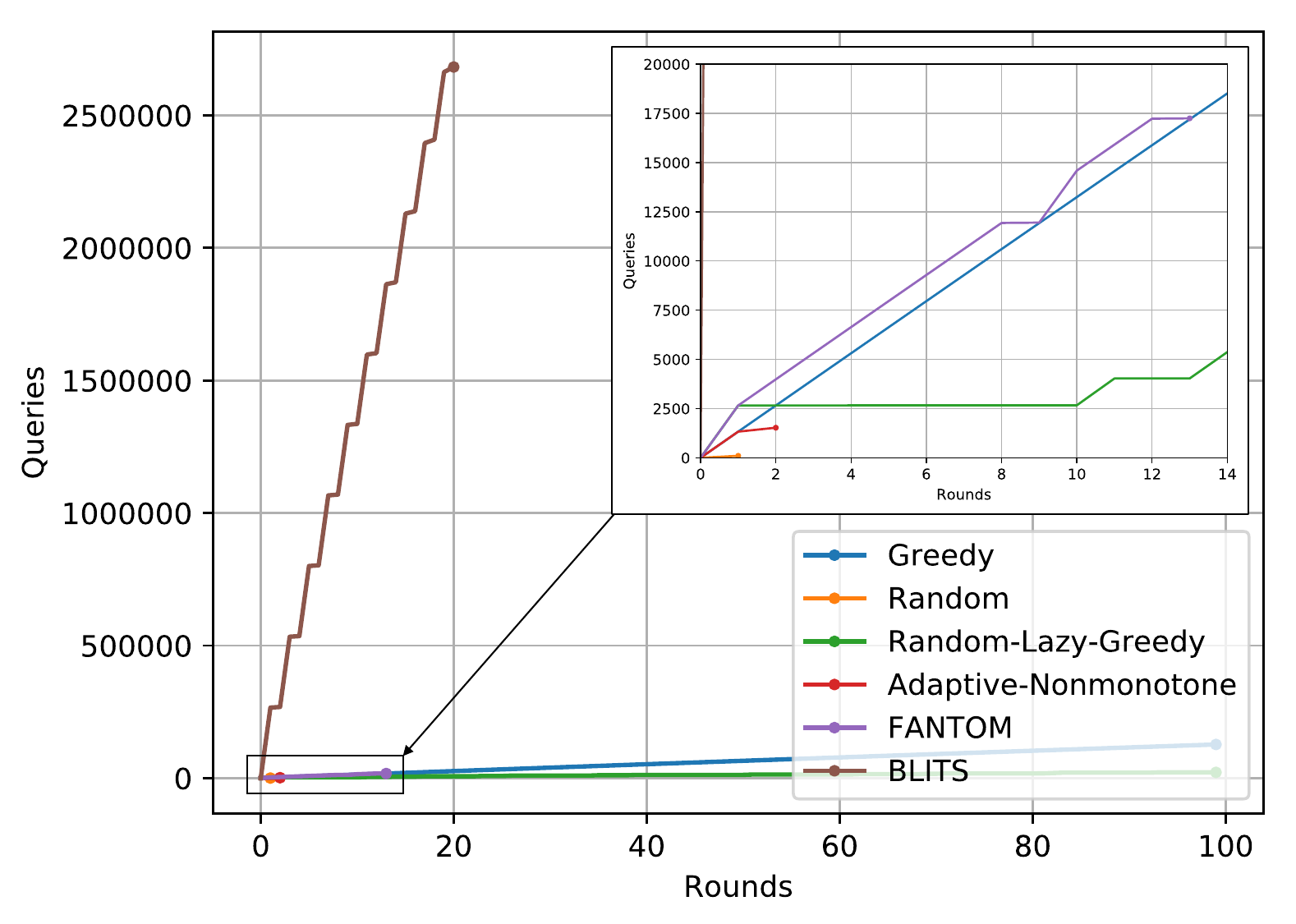}
\caption{Revenue Maximization}
\label{fig:youtube-queries}
\end{subfigure}

\caption{Performance of \NonmonotoneMaximization compared to several benchmarks
for image summarization on the CIFAR-10 dataset,
movie recommendation on the MovieLens 20M dataset,
and revenue maximization on the top 5,000 communities of YouTube.
}
\label{fig:all-plots}
%\vspace{-0.4cm}
\end{figure*}

\textbf{Image Summarization.}
The goal of image summarization is to find
a small, representative subset from a
large collection of images that accurately describes the entire dataset.
The quality of a summary is typically modeled by
two contrasting requirements: coverage and diversity.
Coverage measures the overall representation of the dataset,
and diversity encourages succinctness by penalizing summaries that contain
similar images.
For a collection of images $N$, the objective function
we use for image summarization is
\begin{equation*}
\label{eqn:image-objective}
  f(S) = \sum_{i \in N} \max_{j \in S} s_{i,j}
        - \frac{1}{|N|} \sum_{i \in S} \sum_{j \in S} s_{i,j},
\end{equation*}
where $s_{i,j}$ is the similarity between image $i$ and image~$j$.
The trade-off between coverage and diversity naturally gives rise to
non-monotone submodular functions.
We perform our image summarization experiment 
on the CIFAR-10 test set~\cite{krizhevsky2009learning},
which contains 10,000 $32 \times 32$ color images.
The image similarity $s_{i,j}$ is measured by the cosine similarity
of the 3,072-dimensional pixel vectors for images $i$ and $j$.
Following \citet{balkanski2018non}, we randomly select $500$ images
to be our subsampled ground set since this experiment is throttled by
the number and cost of oracle queries.

We set $k=80$ in \Cref{fig:images-rounds} and
track the progress of the algorithms in each round.
\Cref{fig:images-constraints} compares the solution quality for different
constraints $k \in (20,40,60,80,100)$ and demonstrates that
\NonmonotoneMaximization and \Blits find substantially better solutions
than \Random.
We use $10$ trials for each stochastic algorithm and plot the mean and standard
deviation of the solutions.
We note that \Fantom performs noticeably worse than the others
because it stops choosing elements when their (possibly positive) marginal gain
falls below a fixed threshold.
We give a picture-in-picture plot of the query complexities in
\Cref{fig:images-queries} to highlight the difference in overall cost of the
estimators for \NonmonotoneMaximization and \Blits.

%\subsection{Movie Recommendation}
\textbf{Movie Recommendation.}
Personalized movie recommendation systems aim to provide short, comprehensive
lists of high-quality movies for a user based on the ratings of similar users.
In this experiment, we randomly sample 500 movies from the 
MovieLens 20M dataset~\cite{harper2015movielens}, which contains
20 million ratings for 26,744 movies by 138,493 users. We use
\textsc{Soft-Impute}~\cite{mazumder2010spectral}
to predict the rating vector for each movie via low-rank
matrix completion, and we define the similarity of two movies $s_{i,j}$
as the inner product of the rating vectors for movies $i$ and~$j$.
Following \citet{mirzasoleiman2016fast}, we use the objective function
\begin{equation*}
  f(S) = \sum_{i \in N} \sum_{j \in S} s_{i,j} 
          - \lambda \sum_{i \in S} \sum_{j \in S} s_{i,j},
\end{equation*}
with $\lambda = 0.95$. Note that if $\lambda=1$ we have the cut function.

We remark that experiment is similar to solving max-cut on an
Erd\"os-R\'enyi graph.
In \Cref{fig:movie-rounds} we set $k=200$,
and in \Cref{fig:movie-constraints} we consider $k \in (50,100,150,200,250)$.
The \Greedy algorithm performs moderately better than \Random as the constraint
approaches $k=250$, and all other algorithms except \Fantom are 
sandwiched between these benchmarks.
The query complexities are similar to \Cref{fig:images-queries}, so we
exclude this plot to keep \Cref{fig:all-plots} compact.

%\subsection{Revenue Maximization}
\textbf{Revenue Maximization.}
In this application, our goal is to choose a subset of users in a social
network to advertise a product in order to maximize its revenue.
We consider the top 5,000 communities of the YouTube network~\cite{snapnets}
and subsample the graph by restricting to~25
randomly chosen communities~\cite{balkanski2018non}.
The resulting network has 1,329 nodes and 3,936 edges.
We assign edge weights according to the continuous
uniform distribution $\mathcal{U}(0,1)$, and
we measure influence using the non-monotone function
\begin{equation*}
  f(S) = \sum_{i \in N\setminus S} \sqrt{\sum_{j \in S} w_{i,j}}.
\end{equation*}

In \Cref{fig:youtube-rounds}, we set $k=100$ and observe that
\NonmonotoneMaximization significantly outperforms \Fantom and \Random.
\Cref{fig:youtube-constraints} shows a stratification
of the algorithms for $k \in (20,40,60,80,100)$,
and \Cref{fig:youtube-queries}
is similar to the image summarization experiment.
We note that the inner plot in \Cref{fig:youtube-queries}
shows that for the optimal threshold of \NonmonotoneMaximization, the
number of candidates instantly falls below $3k$ and the algorithm outputs a
random prefix of high-valued elements in the next round.

\section{Conclusions}
\label{sec:conclusion}

We give the first algorithm for maximizing a non-monotone submodular function
subject to a cardinality constraint that achieves a constant-factor
approximation with nearly optimal adaptivity complexity.
The query complexity of this algorithm is also nearly optimal and considerably less than in previous works.
While the approximation guarantee is only $0.039-\varepsilon$,
our empirical study shows that for several real-world applications
\NonmonotoneMaximization
finds solutions that are competitive with benchmarks for
non-monotone submodular maximization
and requires significantly fewer rounds and oracle queries.

\newpage

\section*{Acknowledgements}
We thank the anonymous reviewers for their valuable feedback.
\change{We also especially thank Canh Pham Van for bringing
an error in the ICML 2019 manuscript to our attention,
where we implicitly assumed monotonicity by using an
unmodified version of threshold sampling algorithm in~\citet{fahrbach2019submodular}.}
M.F.\ was supported in part by an NSF
Graduate Research Fellowship
under grant DGE-1650044.
Part of this work was done while he was a summer
intern at Google Research, Z\"urich.

\bibliographystyle{icml2019}
\bibliography{references}

\begin{thebibliography}{47}
\providecommand{\natexlab}[1]{#1}
\providecommand{\url}[1]{\texttt{#1}}
\expandafter\ifx\csname urlstyle\endcsname\relax
  \providecommand{\doi}[1]{doi: #1}\else
  \providecommand{\doi}{doi: \begingroup \urlstyle{rm}\Url}\fi

\bibitem[Agarwal et~al.(2019)Agarwal, Assadi, and
  Khanna]{agarwal2019stochastic}
Agarwal, A., Assadi, S., and Khanna, S.
\newblock Stochastic submodular cover with limited adaptivity.
\newblock In \emph{Proceedings of the Thirtieth Annual ACM-SIAM Symposium on
  Discrete Algorithms}, pp.\  323--342. SIAM, 2019.

\bibitem[Amanatidis et~al.(2021)Amanatidis, Fusco, Lazos, Leonardi,
  Marchetti-Spaccamela, and Reiffenh{\"a}user]{amanatidis2021submodular}
Amanatidis, G., Fusco, F., Lazos, P., Leonardi, S., Marchetti-Spaccamela, A.,
  and Reiffenh{\"a}user, R.
\newblock Submodular maximization subject to a knapsack constraint:
  {C}ombinatorial algorithms with near-optimal adaptive complexity.
\newblock In \emph{International Conference on Machine Learning}, pp.\
  231--242. PMLR, 2021.

\bibitem[Badanidiyuru et~al.(2014)Badanidiyuru, Mirzasoleiman, Karbasi, and
  Krause]{badanidiyuru2014streaming}
Badanidiyuru, A., Mirzasoleiman, B., Karbasi, A., and Krause, A.
\newblock Streaming submodular maximization: Massive data summarization on the
  fly.
\newblock In \emph{Proceedings of the 20th ACM SIGKDD International Conference
  on Knowledge Discovery and Data Mining}, pp.\  671--680. ACM, 2014.

\bibitem[Balkanski \& Singer(2018)Balkanski and Singer]{BS18}
Balkanski, E. and Singer, Y.
\newblock The adaptive complexity of maximizing a submodular function.
\newblock In \emph{Proceedings of the 50th Annual ACM SIGACT Symposium on
  Theory of Computing}, pp.\  1138--1151. ACM, 2018.

\bibitem[Balkanski et~al.(2018)Balkanski, Breuer, and Singer]{balkanski2018non}
Balkanski, E., Breuer, A., and Singer, Y.
\newblock Non-monotone submodular maximization in exponentially fewer
  iterations.
\newblock In \emph{Advances in Neural Information Processing Systems}, pp.\
  2359--2370, 2018.

\bibitem[Balkanski et~al.(2019)Balkanski, Rubinstein, and Singer]{BRS18}
Balkanski, E., Rubinstein, A., and Singer, Y.
\newblock An exponential speedup in parallel running time for submodular
  maximization without loss in approximation.
\newblock In \emph{Proceedings of the Thirtieth Annual ACM-SIAM Symposium on
  Discrete Algorithms}, pp.\  283--302. SIAM, 2019.

\bibitem[Barbosa et~al.(2015)Barbosa, Ene, Nguyen, and Ward]{barbosa2015power}
Barbosa, R., Ene, A., Nguyen, H., and Ward, J.
\newblock The power of randomization: Distributed submodular maximization on
  massive datasets.
\newblock In \emph{International Conference on Machine Learning}, pp.\
  1236--1244. PMLR, 2015.

\bibitem[Barbosa et~al.(2016)Barbosa, Ene, Nguyen, and Ward]{alina2}
Barbosa, R. d.~P., Ene, A., Nguyen, H.~L., and Ward, J.
\newblock A new framework for distributed submodular maximization.
\newblock In \emph{2016 IEEE 57th Annual Symposium on Foundations of Computer
  Science (FOCS)}, pp.\  645--654. IEEE, 2016.

\bibitem[Buchbinder \& Feldman(2019)Buchbinder and
  Feldman]{buchbinder2016constrained}
Buchbinder, N. and Feldman, M.
\newblock Constrained submodular maximization via a nonsymmetric technique.
\newblock \emph{Mathematics of Operations Research}, 44\penalty0 (3):\penalty0
  988--1005, 2019.

\bibitem[Buchbinder et~al.(2014)Buchbinder, Feldman, Naor, and
  Schwartz]{buchbinder2014submodular}
Buchbinder, N., Feldman, M., Naor, J.~S., and Schwartz, R.
\newblock Submodular maximization with cardinality constraints.
\newblock In \emph{Proceedings of the Twenty-Fifth Annual ACM-SIAM Symposium on
  Discrete Algorithms}, pp.\  1433--1452. SIAM, 2014.

\bibitem[Buchbinder et~al.(2016)Buchbinder, Feldman, and
  Schwartz]{buchbinder2016comparing}
Buchbinder, N., Feldman, M., and Schwartz, R.
\newblock Comparing apples and oranges: Query trade-off in submodular
  maximization.
\newblock \emph{Mathematics of Operations Research}, 42\penalty0 (2):\penalty0
  308--329, 2016.

\bibitem[Chekuri \& Quanrud(2019{\natexlab{a}})Chekuri and
  Quanrud]{chekuri2018parallelizing}
Chekuri, C. and Quanrud, K.
\newblock Parallelizing greedy for submodular set function maximization in
  matroids and beyond.
\newblock In \emph{Proceedings of the 51st Annual ACM SIGACT Symposium on
  Theory of Computing}, pp.\  78--89. ACM, 2019{\natexlab{a}}.

\bibitem[Chekuri \& Quanrud(2019{\natexlab{b}})Chekuri and
  Quanrud]{chekuri2019submodular}
Chekuri, C. and Quanrud, K.
\newblock Submodular function maximization in parallel via the multilinear
  relaxation.
\newblock In \emph{Proceedings of the Thirtieth Annual ACM-SIAM Symposium on
  Discrete Algorithms}, pp.\  303--322. SIAM, 2019{\natexlab{b}}.

\bibitem[Chen et~al.(2019)Chen, Feldman, and Karbasi]{chen2018unconstrained}
Chen, L., Feldman, M., and Karbasi, A.
\newblock Unconstrained submodular maximization with constant adaptive
  complexity.
\newblock In \emph{Proceedings of the 51st Annual ACM SIGACT Symposium on
  Theory of Computing}, pp.\  102--113. ACM, 2019.

\bibitem[Chen \& Kuhnle(2022)Chen and Kuhnle]{chen2022practical}
Chen, Y. and Kuhnle, A.
\newblock Practical and parallelizable algorithms for non-monotone submodular
  maximization with size constraint.
\newblock \emph{arXiv preprint arXiv:2009.01947}, 2022.

\bibitem[Das \& Kempe(2008)Das and Kempe]{DK08}
Das, A. and Kempe, D.
\newblock Algorithms for subset selection in linear regression.
\newblock In \emph{Proceedings of the Fortieth Annual ACM Symposium on Theory
  of Computing}, pp.\  45--54. ACM, 2008.

\bibitem[Dueck \& Frey(2007)Dueck and Frey]{dueck2007non}
Dueck, D. and Frey, B.~J.
\newblock Non-metric affinity propagation for unsupervised image
  categorization.
\newblock In \emph{Computer Vision, 2007. ICCV 2007. IEEE 11th International
  Conference on}, pp.\  1--8. IEEE, 2007.

\bibitem[El-Arini \& Guestrin(2011)El-Arini and Guestrin]{el2011beyond}
El-Arini, K. and Guestrin, C.
\newblock Beyond keyword search: {D}iscovering relevant scientific literature.
\newblock In \emph{Proceedings of the 17th ACM SIGKDD International Conference
  on Knowledge Discovery and Data Mining}, pp.\  439--447. ACM, 2011.

\bibitem[Elenberg et~al.(2018)Elenberg, Khanna, Dimakis, and
  Negahban]{elenberg2018restricted}
Elenberg, E.~R., Khanna, R., Dimakis, A.~G., and Negahban, S.
\newblock Restricted strong convexity implies weak submodularity.
\newblock \emph{The Annals of Statistics}, 46\penalty0 (6B):\penalty0
  3539--3568, 2018.

\bibitem[Ene \& Nguyen(2019)Ene and Nguyen]{EN18}
Ene, A. and Nguyen, H.~L.
\newblock Submodular maximization with nearly-optimal approximation and
  adaptivity in nearly-linear time.
\newblock In \emph{Proceedings of the Thirtieth Annual ACM-SIAM Symposium on
  Discrete Algorithms}, pp.\  274--282. SIAM, 2019.

\bibitem[Ene et~al.(2018)Ene, Nguyen, and Vladu]{ene2018parallel}
Ene, A., Nguyen, H.~L., and Vladu, A.
\newblock A parallel double greedy algorithm for submodular maximization.
\newblock \emph{arXiv preprint arXiv:1812.01591}, 2018.

\bibitem[Ene et~al.(2019)Ene, Nguyen, and Vladu]{ene2018submodular}
Ene, A., Nguyen, H.~L., and Vladu, A.
\newblock Submodular maximization with matroid and packing constraints in
  parallel.
\newblock In \emph{Proceedings of the 51st annual ACM SIGACT Symposium on
  Theory of Computing}, pp.\  90--101. ACM, 2019.

\bibitem[Fahrbach et~al.(2019)Fahrbach, Mirrokni, and
  Zadimoghaddam]{fahrbach2019submodular}
Fahrbach, M., Mirrokni, V., and Zadimoghaddam, M.
\newblock Submodular maximization with nearly optimal approximation, adaptivity
  and query complexity.
\newblock In \emph{Proceedings of the Thirtieth Annual ACM-SIAM Symposium on
  Discrete Algorithms}, pp.\  255--273. SIAM, 2019.

\bibitem[Feige et~al.(2011)Feige, Mirrokni, and Vondrak]{feige2011maximizing}
Feige, U., Mirrokni, V.~S., and Vondrak, J.
\newblock Maximizing non-monotone submodular functions.
\newblock \emph{SIAM Journal on Computing}, 40\penalty0 (4):\penalty0
  1133--1153, 2011.

\bibitem[Gharan \& Vondr{\'a}k(2011)Gharan and
  Vondr{\'a}k]{gharan2011submodular}
Gharan, S.~O. and Vondr{\'a}k, J.
\newblock Submodular maximization by simulated annealing.
\newblock In \emph{Proceedings of the Twenty-Second Annual ACM-SIAM Symposium
  on Discrete Algorithms}, pp.\  1098--1116. SIAM, 2011.

\bibitem[Harper \& Konstan(2016)Harper and Konstan]{harper2015movielens}
Harper, F.~M. and Konstan, J.~A.
\newblock The movielens datasets: {H}istory and context.
\newblock \emph{ACM Transactions on Interactive Intelligent Systems (TIIS)},
  5\penalty0 (4):\penalty0 19, 2016.

\bibitem[Hartline et~al.(2008)Hartline, Mirrokni, and Sundararajan]{HMS08}
Hartline, J., Mirrokni, V., and Sundararajan, M.
\newblock Optimal marketing strategies over social networks.
\newblock In \emph{Proceedings of the 17th international conference on World
  Wide Web}, pp.\  189--198. ACM, 2008.

\bibitem[Kazemi et~al.(2018)Kazemi, Zadimoghaddam, and
  Karbasi]{kazemi2018scalable}
Kazemi, E., Zadimoghaddam, M., and Karbasi, A.
\newblock Scalable deletion-robust submodular maximization: Data summarization
  with privacy and fairness constraints.
\newblock In \emph{International Conference on Machine Learning}, pp.\
  2549--2558. PMLR, 2018.

\bibitem[Khanna et~al.(2017)Khanna, Elenberg, Dimakis, Negahban, and
  Ghosh]{KEDNG17}
Khanna, R., Elenberg, E.~R., Dimakis, A.~G., Negahban, S., and Ghosh, J.
\newblock Scalable greedy feature selection via weak submodularity.
\newblock In \emph{Proceedings of the 20th International Conference on
  Artificial Intelligence and Statistics}, pp.\  1560--1568, 2017.

\bibitem[Krizhevsky \& Hinton(2009)Krizhevsky and
  Hinton]{krizhevsky2009learning}
Krizhevsky, A. and Hinton, G.
\newblock Learning multiple layers of features from tiny images.
\newblock Technical report, Citeseer, 2009.

\bibitem[Kuhnle(2021)]{kuhnle2021nearly}
Kuhnle, A.
\newblock Nearly linear-time, parallelizable algorithms for non-monotone
  submodular maximization.
\newblock In \emph{Proceedings of the AAAI Conference on Artificial
  Intelligence}, volume~35, pp.\  8200--8208, 2021.

\bibitem[Kumar et~al.(2015)Kumar, Moseley, Vassilvitskii, and Vattani]{KMVV13}
Kumar, R., Moseley, B., Vassilvitskii, S., and Vattani, A.
\newblock Fast greedy algorithms in mapreduce and streaming.
\newblock \emph{ACM Transactions on Parallel Computing (TOPC)}, 2\penalty0
  (3):\penalty0 14, 2015.

\bibitem[Lattanzi et~al.(2011)Lattanzi, Moseley, Suri, and
  Vassilvitskii]{spaa-LMSV11}
Lattanzi, S., Moseley, B., Suri, S., and Vassilvitskii, S.
\newblock Filtering: {A} method for solving graph problems in mapreduce.
\newblock In \emph{Proceedings of the Twenty-Third Annual ACM Symposium on
  Parallelism in Algorithms and Architectures}, pp.\  85--94. ACM, 2011.

\bibitem[Lee et~al.(2010)Lee, Mirrokni, Nagarajan, and Sviridenko]{LMNS10}
Lee, J., Mirrokni, V.~S., Nagarajan, V., and Sviridenko, M.
\newblock Maximizing nonmonotone submodular functions under matroid or knapsack
  constraints.
\newblock \emph{SIAM Journal on Discrete Mathematics}, 23\penalty0
  (4):\penalty0 2053--2078, 2010.

\bibitem[Leskovec \& Krevl(2014)Leskovec and Krevl]{snapnets}
Leskovec, J. and Krevl, A.
\newblock {SNAP Datasets}: {Stanford} large network dataset collection.
\newblock \url{http://snap.stanford.edu/data}, June 2014.

\bibitem[Liu \& Vondr{\'{a}}k(2019)Liu and Vondr{\'{a}}k]{liu2018submodular}
Liu, P. and Vondr{\'{a}}k, J.
\newblock Submodular optimization in the {M}ap{R}educe model.
\newblock In \emph{2nd Symposium on Simplicity in Algorithms}, volume~69, pp.\
  18:1--18:10. Schloss Dagstuhl - Leibniz-Zentrum f{\"{u}}r Informatik, 2019.

\bibitem[Mazumder et~al.(2010)Mazumder, Hastie, and
  Tibshirani]{mazumder2010spectral}
Mazumder, R., Hastie, T., and Tibshirani, R.
\newblock Spectral regularization algorithms for learning large incomplete
  matrices.
\newblock \emph{Journal of Machine Learning Research}, 11:\penalty0 2287--2322,
  2010.

\bibitem[Mirrokni \& Zadimoghaddam(2015)Mirrokni and
  Zadimoghaddam]{MirrokniZadim2015}
Mirrokni, V. and Zadimoghaddam, M.
\newblock Randomized composable core-sets for distributed submodular
  maximization.
\newblock In \emph{Proceedings of the Forty-Seventh Annual ACM Symposium on
  Theory of Computing}, pp.\  153--162. ACM, 2015.

\bibitem[Mirzasoleiman et~al.(2013)Mirzasoleiman, Karbasi, Sarkar, and
  Krause]{nips13}
Mirzasoleiman, B., Karbasi, A., Sarkar, R., and Krause, A.
\newblock Distributed submodular maximization: Identifying representative
  elements in massive data.
\newblock In \emph{Advances in Neural Information Processing Systems}, pp.\
  2049--2057, 2013.

\bibitem[Mirzasoleiman et~al.(2016)Mirzasoleiman, Badanidiyuru, and
  Karbasi]{mirzasoleiman2016fast}
Mirzasoleiman, B., Badanidiyuru, A., and Karbasi, A.
\newblock Fast constrained submodular maximization: Personalized data
  summarization.
\newblock In \emph{International Conference on Machine Learning}, pp.\
  1358--1367. PMLR, 2016.

\bibitem[Mohri et~al.(2018)Mohri, Rostamizadeh, and
  Talwalkar]{mohri2018foundations}
Mohri, M., Rostamizadeh, A., and Talwalkar, A.
\newblock \emph{Foundations of Machine Learning}.
\newblock MIT Press, 2018.

\bibitem[Nemhauser et~al.(1978)Nemhauser, Wolsey, and
  Fisher]{Nemhauser_Wolsey_Fisher78}
Nemhauser, G.~L., Wolsey, L.~A., and Fisher, M.~L.
\newblock An analysis of approximations for maximizing submodular set
  functions.
\newblock \emph{Mathematical Programming}, 14\penalty0 (1):\penalty0 265--294,
  1978.

\bibitem[Norouzi-Fard et~al.(2018)Norouzi-Fard, Tarnawski, Mitrovic, Zandieh,
  Mousavifar, and Svensson]{norouzi2018beyond}
Norouzi-Fard, A., Tarnawski, J., Mitrovic, S., Zandieh, A., Mousavifar, A., and
  Svensson, O.
\newblock Beyond 1/2-approximation for submodular maximization on massive data
  streams.
\newblock In \emph{International Conference on Machine Learning}, pp.\
  3829--3838. PMLR, 2018.

\bibitem[Qian \& Singer(2019)Qian and Singer]{qian2019fast}
Qian, S. and Singer, Y.
\newblock Fast parallel algorithms for statistical subset selection problems.
\newblock In \emph{Advances in Neural Information Processing Systems},
  volume~32, 2019.

\bibitem[Simon et~al.(2007)Simon, Snavely, and Seitz]{simon2007scene}
Simon, I., Snavely, N., and Seitz, S.~M.
\newblock Scene summarization for online image collections.
\newblock In \emph{2007 IEEE 11th International Conference on Computer Vision},
  pp.\  1--8. IEEE, 2007.

\bibitem[Sipos et~al.(2012)Sipos, Swaminathan, Shivaswamy, and
  Joachims]{sipos2012temporal}
Sipos, R., Swaminathan, A., Shivaswamy, P., and Joachims, T.
\newblock Temporal corpus summarization using submodular word coverage.
\newblock In \emph{Proceedings of the 21st ACM International Conference on
  Information and Knowledge Management}, pp.\  754--763. ACM, 2012.

\bibitem[Tschiatschek et~al.(2014)Tschiatschek, Iyer, Wei, and
  Bilmes]{tschiatschek2014learning}
Tschiatschek, S., Iyer, R.~K., Wei, H., and Bilmes, J.~A.
\newblock Learning mixtures of submodular functions for image collection
  summarization.
\newblock In \emph{Advances in Neural Information Processing Systems}, pp.\
  1413--1421, 2014.

\end{thebibliography}

\appendix
\newpage
\onecolumn{
\section{Missing Analysis from \Cref{sec:preliminaries}}

\change{
\NonmonotoneThresholdSamplingAlg*

\begin{proof}
We start by proving that the adaptivity complexity of \NMTS is $O(\log(n/\delta)/\varepsilon)$.
The number of rounds is $r=O(\log_{(1-\varepsilon)^{-1}}(n/\delta))$ by construction,
and there are polynomially-many queries in each round,
all of which are independent and only rely on the state of $S$ at the
beginning of the round.
Finally, updating~$S'$ on Line~\ref{step:post_filter}
requires one adaptive round since each post-filtering step happens
with respect to $S$ and prefixes of a fixed permutation.

\paragraph{Property 1.}
The expected $O(n/\varepsilon)$ query complexity
follows from~\citet[Lemma~3.2]{fahrbach2019submodular}
for the original \AdaptiveSampling algorithm
since this part of their analysis holds for general submodular functions.
The only change we need to consider for \NMTS
is the number of queries used in the construction of
$S'$ on Line~\ref{step:post_filter}.
For each $x_i \in T$, the algorithm makes two oracle calls to
decide if $\Delta(x_i,S \cup\{x_1,\dots,x_{i-1}\}) \ge \tau$,
so the total query complexity of Line~\ref{step:post_filter} is $2|S| = O(n)$.
Therefore, the expected query complexity of the algorithm is $O(n/\varepsilon)$.

\paragraph{Property 2.}
If $|S| < k$ when the algorithm terminates,
then \NMTS did not break on Line~17.
Therefore, the algorithm either (a) breaks early on Line~7,
in which the property holds,
or (b) finishes all $r$ rounds.
In the case where all $r$ rounds finish and $|S| < k$,
it follows from the proof of
\citet[Lemma 3.6]{fahrbach2019submodular} that
\begin{align*}
    \Prob{}{|A_r| \ge c k \mid Z}
    &\le
    \Prob{}{|A_r| \ge 1 \mid Z} \\
    &\le
    \delta / 2,
\end{align*}
since we always take $c \ge 1$.
Therefore, with probability at least $1 - \delta/2$,
the number of remaining candidates is $|A| < c k$.

\paragraph{Property 3.}
For each round $\ell \in [r]$, let:
\begin{itemize}
\item $A_{\ell}$ be the value of $A$ after the filtering step in Line~6
\item $T_{\ell}$ be the value of $T$
\item $S_{\ell}$ be the value of $S$ before the update in Line~16
\end{itemize}
We use the round-specific values for these random sets to prove the next three properties.

Let $T_{\ell}' \subseteq T_{\ell}$ denote the set of elements in $T_{\ell}$ that survive the post-filtering step
in~Line~\ref{step:post_filter}.
We want to show that
\[
    \Exp{}{|T'_{\ell}| \mid Z} \geq (1 - \varepsilon) \Exp{}{|T_{\ell}| \mid Z}.
\]

%Moved from SODA paper
In the for loop on Line~9, we go over iterations $i = 0, 1, \dots, m$. 
Let $t^* = \min\{\lfloor (1+\hat\varepsilon)^i \rfloor, |A|\}$ be the final value we 
set for~$t$ in the last iteration of the for loop on Lines~9--12. 
In other words, we break on Line~12 during iteration $0 \leq i \leq m$,
or the loop finishes.
If $t^* = 1$, then we have $|T'_{\ell}| = |T_\ell| = 1$ by the definition of $A_\ell$.

Therefore, assume that $t^* \ge 2$.
The size $t \le t^*$ in the previous iteration, where
\[
    t = \floor{(1+\hat\varepsilon)^{i-1}} \le \floor{(1+\hat\varepsilon)^{i}} = t^*,
\]
satisfies $\Exp{}{I_t \mid Z} \ge 1 - 2\hat\varepsilon$ since the algorithm
did not break on Line~11 in iteration $i-1$.
Thus,
$\Exp{}{I_{t'} \mid Z} \ge 1 - 2\hat\varepsilon$
for any $0 \le t' \le t$.
Note that we do not necessarily sample $t^*$ elements on Line~13 since the sample size is capped at $k-|S|$.

It follows from a decomposition of $|T'_{\ell}|$
into indicator random variables that
\begin{align} \nonumber
    \Exp{}{|T_{\ell}'| \mid Z, |T_{\ell}| = b}
    &=
    \sum_{t'=0}^{b-1} \E[I_{t'} \mid Z] \\ \nonumber
    &\ge
    \sum_{t'=0}^{\min\{t, b-1\}} \E[I_{t'} \mid Z] \\ \nonumber
    &\ge
    \parens*{\min\{t + 1, b\}} (1-2\hat\varepsilon) \\ \nonumber
    &\ge
    \frac{b}{1+\hat\varepsilon} (1-2\hat\varepsilon) \\
    &\ge
    (1-\varepsilon) \Exp{}{|T_{\ell}| \mid Z, |T_{\ell}| = b}. \label{eq:LowerBoundT'Condition-on-size-of-T}
\end{align}

To prove the second to last inequality, we get help from two observations:
$t^* \le (1+\hat\varepsilon)(t + 1)$, and $b \le t^*$.
Therefore, $\min\{t + 1, b\}$ is at least $b/(1+\hat\varepsilon)$.

\iffalse
It follows from a decomposition of $|T'_{\ell}|$
into indicator random variables
and the fact that $t < b$ that
\begin{align*}
    \Exp{}{|T_{\ell}'| \mid Z, |T_{\ell}| = b}
    &=
    \sum_{t'=0}^{b-1} \E[I_{t'} \mid Z] \\
    &\ge
    \sum_{t'=0}^{t} \E[I_{t'} \mid Z] \\
    &\ge
    \parens*{t + 1} (1-2\hat\varepsilon) \\
    &\ge
    (1-\varepsilon) \Exp{}{|T_{\ell}| \mid Z, |T_{\ell}| = b}.
    \todo{Last inequality is not true yet?}
\end{align*}

We note that $t$ is at least $t^*/(1+\hat\varepsilon)$, and therefore it is also at least $\ceil{b/(1+\hat\varepsilon)}$ as it appears in the first inequality. 
\fi

Since Inequality~\eqref{eq:LowerBoundT'Condition-on-size-of-T} holds for any value of $|T_{\ell}| = b$, we can summarize all of them in the following form:
\[
    \Exp{}{|T_{\ell}'| \mid Z} \geq (1-\varepsilon) \Exp{}{|T_{\ell}| \mid Z}.
\]

Putting everything together,
\begin{align*}
    \Exp{}{|S'| \mid Z}
    &=
    \sum_{i=1}^r \Exp{}{|T_i'| \mid Z} \\
    &\ge
    (1 - \varepsilon) \sum_{i=1}^r \Exp{}{|T_i| \mid Z} \\
    &=
    (1 - \varepsilon) \Exp{}{|S| \mid Z}.
\end{align*}

\paragraph{Property 4.}
For each $T_\ell$, let the shuffled elements on~Line~14 of \Cref{alg:sampling} be
$(x_{\ell,1},x_{\ell,2},\dots,x_{\ell,|T_\ell|})$.
It follows from submodularity and the definition of $S'$ that
\begin{align*}
  f(S')
  &=
  \sum_{\ell=1}^{r}
  \sum_{j=1}^{|T_\ell|}
    \Delta(x_{\ell,j}, S' \cap (S_\ell \cup \{x_{\ell,1},\dots,x_{\ell,j-1}\}))
    \cdot
    \mathbf{1}_{S'}(x_{\ell,j}) \\
  &\ge
  \sum_{\ell=1}^{r}
  \sum_{j=1}^{|T_\ell|}
    \Delta(x_{\ell,j}, S_\ell \cup \{x_{\ell,1},\dots,x_{\ell,j-1}\})
    \cdot
    \mathbf{1}_{S'}(x_{\ell,j}) \\
  &\ge
  \sum_{\ell=1}^{r}
  \sum_{j=1}^{|T_\ell|}
    \tau
    \cdot
    \mathbf{1}_{S'}(x_{\ell,j}) \\
  &=
    \tau \cdot |S'|,
\end{align*}
where $\mathbf{1}_{B}(x)$ is the indicator function for a set $B$ defined as
$\mathbf{1}_{B}(x) = 1$ if $x \in B$ and
$\mathbf{1}_{B}(x) = 0$ if $x \not\in B$.

\paragraph{Property 5.}
%Let $A_i$ and $T_i$ denote the values of the random sets
%$A$ and $T$ in round $i$.
For any $x \in N$,
let $X_\ell$ be an indicator random variable for the event $x \in T_\ell$.
It follows that
\begin{align*}
  \Prob{}{x \in S}
  &=
  \sum_{\ell=1}^r \Exp{}{X_\ell}
  \le
  \sum_{\ell=1}^r \Exp{}{\frac{|T_\ell|}{|A_\ell|}}
  \le
  \frac{1}{\morteza{c} k}
  \sum_{\ell=1}^r \Exp{}{|T_\ell|}
  =
  \frac{1}{\morteza{c} k} \Exp{}{|S|}
  \le
  \frac{1}{\morteza{c} k} \cdot k
  =
  \frac{1}{\morteza{c}}.
\end{align*}
\morteza{We note that one can also upper bound $\Prob{}{x \in S \mid Z}$ as follows. If we condition on $Z$, the choices the algorithm makes in rounds $1,2, \dots, \ell$, and the size of $T_{\ell}$, then we can still upper bound the expected value of $X_{\ell}$ similarly. 
Summing over all possible outcomes of the algorithm and also the values of $|T_{\ell}|$ over rounds $\ell = 1,2,\dots, r$, we achieve the overall upper bound of $1/c$ for $\Prob{}{x \in S \mid Z}$.}
Finally, we have
\begin{align*}
    \Prob{}{x \in S'}
    =
    \Prob{}{x \in S' \mid x \in S} \Pr(x \in S)
    \le
    \Pr(x \in S)
    \le \frac{1}{\morteza{c}},
\end{align*}
which completes the proof.
\end{proof}
}

\unconstrainedMaximizationAlg*

\begin{proof}
First assume that $\OPT_{A} > 0$. We start by bounding the individual failure
probability
\[
  \Prob{}{f\parens*{R_i} \le (1/4 - \varepsilon)\OPT_{A}}
  \le
  \frac{3}{3+4\varepsilon}.
\]
By \Cref{lem:unconstrained-maximization} we have
$\Exp{}{f(R_i)} \ge (1/4)\OPT_{A}$.
Using an analog of Markov's inequality to upper bound
$\E[f(R_i)]$, it follows that
\begin{align*}
  \frac{\OPT_{A}}{4} \le \Exp{}{f\parens*{R_i}}
  \le p \parens*{\frac{1}{4} - \varepsilon}\OPT_{A} + (1-p)\OPT_{A}.
\end{align*}
Therefore, we must have $p \le 3/(3+4\varepsilon)$.
Since the subsets $R_i$ are chosen independently, our choice of $t$ gives us a
total failure probability of
\begin{align*}
  \Prob{}{f(S) \le (1/4-\varepsilon)\OPT_{A}}
    &= \prod_{i=1}^t \Prob{}{f\parens{R_i} \le (1/4 - \varepsilon) \OPT_{A}}\\
    &\le \parens*{\frac{3}{3+4\varepsilon}}^t\\
    &\le \delta.
\end{align*}
This completes the proof that with probability at least
$1-\delta$ we have $f(S) \ge (1/4 - \varepsilon)\OPT_{A}$.
To prove the adaptivity complexity, notice that all subsets $R_i$ can be
generated and evaluated at once in parallel, hence the need for only one
adaptive round. For the query complexity, we use the 
inequality $\log(1+(4/3)\varepsilon) \ge 2\varepsilon/3$, which holds for all
$\varepsilon \le 1/4$.
\end{proof}

%\subsection{Downsampling Submodular Sets}

\section{Missing Analysis from \Cref{sec:nonmonotone}}
\label{app:nonmonotone}

\downsample*

\begin{proof}
Fix an ordering $x_1, x_2,\dots, x_{|S|}$ on the elements in $S$.
Expanding the expected value $\E[f(T)]$ and using submodularity,
it follows that
\begin{align*}
  \Exp{}{f(T)} &= \frac{1}{\binom{|S|}{k}} \sum_{R \in \change{\binom{S}{k}}}
    \sum_{x_i \in R} \Delta\parens*{x_i, \set*{x_1, x_2, \dots, x_{i-1}} \cap R}\\
  &\ge
  \frac{1}{\binom{|S|}{k}} \sum_{R \in \change{\binom{S}{k}}}
    \sum_{x_i \in R} \Delta\parens*{x_i, \set*{x_1, x_2, \dots, x_{i-1}}}\\
  &= \frac{1}{\binom{|S|}{k}} \sum_{i=1}^{|S|}  \binom{|S|-1}{k-1}
    \Delta\parens*{x_i, \set*{x_1,x_2,\dots,x_{i-1}}}\\
  &= \frac{k}{|S|} \cdot f(S),
\end{align*}
which completes the proof.
\end{proof}

\begin{lemma}
\label{lemma:independent-of-A}
For any \change{sets $S, A \subseteq N$} and optimal solution $S^*$, if
$S_2^* = S^* \setminus A$, then
\[
  f\parens*{S_2^*} - f\parens*{S_2^* \cup S} \le f\parens*{S^*} - f\parens*{S^* \cup S}.
\]
\end{lemma}
\begin{proof}
It is equivalent to show that
\begin{align*}
  f(S_2^*) + f(S^* \cup S) \le f(S^*) + f(S_2^* \cup S).
\end{align*}
For any sets $X, Y \subseteq N$, we have 
$f(X \cap Y) + f(X\cup Y) \le f(X) + f(Y)$ by the definition of submodularity.
It follows that
\begin{align*}
  f\parens*{S^* \cap \parens*{S_2^* \cup S}}
  + f\parens*{S^* \cup S}
  \le
  f(S^*) + f(S^*_2 \cup S).
\end{align*}
Therefore, it suffices to instead show that
\begin{equation}
\label{eqn:goal}
  f(S_2^*) \le f(S^* \cap (S_2^* \cup S)).
\end{equation}
Let $S_1^* = S^* \cap A$ and write
\begin{align*}
  S^* \cap \parens*{S_2^* \cup S}
    &= \parens*{S^* \cap S_2^*} \cup \parens*{S^* \cap S}\\
    &= S_2^* \cup \parens*{S_1^* \cap S}.
\end{align*}
Next, fix an ordering $x_1,x_2,\dots,x_\ell$ on the elements in $S^*$.
Summing the consecutive marginal gains of the elements in
the set $S_1^* \cap S$ according to this order gives
\begin{align}
\label{eqn:opt-marginal-gains}
  f\parens*{S_2^* \cup \parens*{S_1^* \cap S}}
  = f(S_2^*) + \sum_{x_1,\dots,x_\ell \in S_1^* \cap S} \Delta\parens*{x_i, S_2^* \cup \set*{x_1,\dots,x_{i-1}}}.
\end{align}
We claim that each marginal contribution in \Cref{eqn:opt-marginal-gains} is
nonnegative.
Assume for contradiction this is not the case.
Let $x^* \in S_1^* \cap S$ be the first element violating this property,
and let $x_{-1}^*$ be the previous element
according to the ordering.
By submodularity, 
\begin{align*}
  0 >
  \Delta\parens*{x^*, S_2^* \cup \bigcup_{x_1,\dots,x_{-1}^* \in S_1^*\cap S} \set*{x_i}}
  \ge
  \Delta\parens*{x^*, S_2^* \cup \bigcup_{x_1,\dots,x_{-1}^* \in S_1^*} \set*{x_i}},
\end{align*}
which implies
$f(S^* \setminus \{x^*\}) > f(S^*) = \OPT$, a contradiction.
Therefore, the inequality in \Cref{eqn:goal} is true, as desired.
\end{proof}

\section{Implementation Details from \Cref{sec:experiments}}
We set $\varepsilon = 0.25$ for all of
the algorithms except \RandomLazyGreedy, which we run with $\varepsilon = 0.01$.
Since some of the algorithms require a guess of $\OPT$,
we adjust~$\varepsilon$ accordingly and fairly.
We remark that all algorithms give reasonably similar results for any
$\varepsilon \in [0.05, 0.50]$.
We set the number of queries to be $100$ for the estimators in
\NonmonotoneMaximization and \Blits,
although for the theoretical guarantees these should be
$\Theta(\log(n)/\varepsilon^2)$ and
$\Theta(\OPT^2 \log(n)/\varepsilon^2)$, respectively.
For context, the experiments in \citet{balkanski2018non} set the number of
samples per estimate to be 30.
Last, we set the number of outer rounds for \Blits
to be $10$, which also matches \citet{balkanski2018non} since the
number needed for provable guarantees is
$r = 20 \log_{1+\varepsilon/2}(n)/\varepsilon$, which
is too large for these datasets.

}

\end{document}